\documentclass[format=manuscript, screen=true, review=false]{acmart}

\usepackage[utf8]{inputenc}
\usepackage[T1]{fontenc}
\usepackage{amsmath}
\usepackage{amsfonts}
\usepackage{graphicx}
\usepackage{color}
\usepackage{url}
\usepackage[normalem]{ulem}

\usepackage{tikz}
\usetikzlibrary{matrix,arrows,patterns,positioning,trees,calc,fit,decorations}
\usetikzlibrary{external}
\usepackage{adjustbox}

\newcommand{\NoShow}[1]{}

\newcommand{\response}[1]{{#1}}
\newcommand{\responsesout}[1]{{\color{teal} \sout{#1}}}

\usepackage{pgfplots}
\pgfplotsset{every axis/.append style={
		scaled y ticks = false,
		scaled x ticks = false,
		y tick label style={/pgf/number format/.cd, fixed, fixed zerofill,
			int detect,1000 sep={\;},precision=3, /tikz/.cd},
		x tick label style={/pgf/number format/.cd, fixed, fixed zerofill,
			int detect, 1000 sep={},precision=3, /tikz/.cd}
	}
}
\pgfplotsset{compat=newest}
\newcommand{\digits}{D}
\newcommand{\gap}{G}
\newcommand{\period}{\mbox{\huge {.}}}
\newcommand{\gemm}{GEMM}
\newcommand{\machepsnew}{\widetilde \epsilon_{\rm mach}}
\newcommand{\macheps}{\epsilon_{\rm mach}}
\newcommand{\machepsdd}{\widehat \epsilon_{\rm mach}}

\newcommand{\specific}[1]{{\color{blue!50!} #1}}
\newcommand{\remove}[1]{{\sout{#1}}}
\renewcommand{\remove}[1]{}

\NoShow{
\acmJournal{TOMS}
\acmVolume{0}
\acmNumber{0}
\acmArticle{0}
\acmYear{2023}
\acmMonth{0}

\setcopyright{acmlicensed}

\acmDOI{xx.xxxx/xxxxxxx}

}

\setcopyright{none}

\author{Devangi N. Parikh}
\email{dnp@cs.utexas.edu}
\affiliation{%
	\department{Oden Institute for Computational Engineering \& Sciences}
	\department{Department of Computer Science}
	\institution{The University of Texas at Austin}
	\city{Austin}
	\country{USA}
}
\author{Robert A. van de Geijn}
\email{rvdg@cs.utexas.edu}
\affiliation{%
	\department{Oden Institute for Computational Engineering \& Sciences}
	\department{Department of Computer Science}
	\institution{The University of Texas at Austin}
	\city{Austin, TX}
	\country{USA}
}
\author{Greg M. Henry}
\email{greg.henry@intel.com}
\affiliation{
	\institution{Intel Corporation}
	\city{Hillsboro, OR}
	\country{USA}
}

	\title{Cascading \gemm: High Precision from Low Precision}

		
	\terms{Algorithms; Performance}
	

\newcommand{\greg}[1]{{\color{red} #1}}
\newcommand{\devangi}[1]{{\color{cyan} #1}}
\newcommand{\robert}[1]{{\color{blue} #1}}

\begin{document}
	
\begin{abstract}
	This paper lays out insights and opportunities for implementing higher-precision matrix-matrix multiplication (\gemm) in terms of lower-precision high-performance \gemm.
	The driving case study approximates double-double precision (FP64x2) \gemm\ in terms of double precision (FP64) \gemm, leveraging how the BLAS-like Library Instantiation Software (BLIS) framework refactors the Goto Algorithm. 
	With this, it is shown how approximate FP64x2 \gemm\ accuracy can be cast  in terms of ten ``cascading'' FP64 \gemm s.  Promising results from preliminary performance and accuracy experiments are reported.  The demonstrated techniques open up new research
	directions for more general cascading of higher-precision computation in terms of lower-precision computation for \gemm-like functionality.
\end{abstract}
\maketitle

\section{Introduction}

The advent of processors with new floating point precisions, including various low precision formats (like half-precision), has raised the question of how to leverage the higher performance of low-precision arithmetic into high performance for high-precision arithmetic~\cite{Bailey,FaHiLoMaMi22,HenHeiTan2019}.
In this paper, we gain insight by  focusing on one end of the spectrum:  the implementation of double-double precision (FP64x2) matrix-matrix multiplication (\gemm), $ C := A B $, in terms of high-performance double precision (FP64) \gemm.
This is an interesting study because currently precisions higher than FP64 are not supported by contemporary CPUs in hardware yet is of interest to some applications~\cite{DeHiLiRi08}.
We then use this case study to discuss a range of new directions that can be explored along the same lines (e.g., casting single precision in terms of half precision).

Our community demands high performance from its \gemm\ implementations~\cite{BLAS3}.  This is for good reasons: in one form or another, much computation in scientific computing and, more recently, machine learning can be cast in terms of a \gemm-like computations (though typically for machine learning these computations are done in lower precision).  Thus, understanding high-performing \gemm\ contributes to high-performing higher-level functionality. For this reason, proposing new techniques should be accompanied by a demonstration of how to achieve performance.

A current framework for effectively implementing \gemm\  functionality is the BLAS-like Library Instantiation Software (BLIS)~\cite{BLIS4, BLIS3,BLIS5,BLIS2,BLIS1}.  Not only does BLIS support high performance that is portable to a variety of contemporary CPUs, but it also allows \gemm-like functionality to be quickly realized.  Key to this is how the framework isolates architecture-specific code and parameters, and how submatrices are packed for data locality.
The principles that underlie BLIS have been exploited in a number of situations, including for tensor contraction~\cite{tblis_sisc}, practical Strassen's algorithms~\cite{7967156,Huang:2016:SAR:3014904.3014983}, k-Nearest Neighbor computations~\cite{Yu:2015:POK:2807591.2807601},
complex \gemm\ in terms of real \gemm~\cite{BLIS6}, and mixed-precision/mixed-domain \gemm~\cite{BLIS7}.
The present work similarly exploits those techniques.

\response{
This paper contributes:
\begin{itemize}
\item
A strategy for cascading FP64x2 matrices into multiple FP64 matrices that enables \responsesout{the} \response{approximating} \responsesout{of} FP64x2 \gemm\ with ten high-performance FP64 \gemm s. This is a factor two better than previous methods.
\item 
An analysis that gives forward error bounds for the proposed method and compares these to those for FP64x2 computation.
\item
A prototype implementation that demonstrates how to make the proposed scheme practical that 
\begin{itemize}
    \item 
achieves high performance \response{(most of the work and time is GEMM-based and our efficiency is superior to prior efforts)} and accuracy \response{that mostly rivals FP64x2 \gemm,}
\item
requires a similar memory footprint to what is already used in a high performance FP64 \gemm\ implementation, 
\item
leverages existing kernels for FP64 \gemm\ for portable high performance, 
\item 
performs
all higher-order $O(n^3)$  and most second-order $O(n^2)$ work in terms of FP64 arithmetic, leaving only some computation to be performed in FP64x2 arithmetic, and
\response{\item can cheaply detect when the accuracy results are suspicious because of ``catastrophic cancellation.'' }
\end{itemize}

    \item 
    A discussion of new opportunities, including:
    \begin{itemize}
    \item
    how cascading can more generally support high precision \gemm\ in term of lower precision \gemm,  
    \item 
    opportunities for leveraging GPUs,
    \item 
    benefits of scaling/balancing methods,
    \item 
    the potential of hardware support for cascading, and
    \item 
    how to support other level-3 BLAS functionality.

    \end{itemize}
\end{itemize}
Together, these suggest the approach shows merit as a compromise between performance and accuracy.
}

\begin{table}[tb!]
\begin{center}
\begin{tabular}{| c | p{4.2in} | }
\hline
Symbol & 
\multicolumn{1}{c|}{
Comment
} \\ \hline \hline
$ \widehat ~ $ &
Quantity/variable associated with FP64x2 
number.
\\
\hline
$ \widetilde ~ $ &
Quantity/variable associated with cascaded number.  \\
\hline
$ A, B, C  $ & Matrices (upper case Roman letters). 
\\
\hline
$ \chi, \psi $ &
Scalars (lower case Greek letters).
\\
\hline
$ x, y $ &
(Column) vectors (lower case Roman letters).
\\
\hline
$m,n$ & Number of rows and columns in $C$, respectively.
\\
\hline
$\vec \jmath$, $ J$ & Vector and matrix of all ones of appropriate size, respectively.
\\
\hline
$k$ & Number of columns in $A$ and number of rows of $B$ (inner dimension).
\\
\hline
$\beta_{i}$ & the $i$-th bit the mantissa of  a floating point number. \\
\hline
$ c_0, c_1, c_2 $ &
Number of digits in splits $ 0, 1, 2 $. 
$ c_0 = D_0 $,
$ c_1 = D_1 - D_0 $,
$ c_2 = D_2 - D_1 $. \\
\hline
$ D $ & \# binary digits in mantissa \specific{of  a FP64}.
\mbox{\specific{$ D = 53 $ for FP64}.}
\\ \hline
$D_0, D_1, D_2 $ &
Boundaries between digit ranges for each split of the cascaded number. \\
\hline
$ D_A $, $ D_B $
&
$ D_A = 
{\rm diag}( 
\max_j \vert \alpha_{0,j} \vert ,
\max_j \vert \alpha_{1,j} \vert, \ldots  ) $, \\
&
$ D_B= 
{\rm diag}( 
\max_i \vert \beta_{i,0} \vert,
\max_j \vert \alpha_{i,1} \vert, \ldots  ) $\\ 
\hline
$ \delta\!\chi,
\delta\!x, \delta\!y,
\Delta\!A, \Delta\!B
$ &
Error in $ \chi, x, y, A, B $.
\\
\hline
$ \epsilon_{\rm mach} $,
$ \widehat \epsilon_{\rm mach} $,
$ \widetilde \epsilon_{\rm mach} $
&
Machine epsilon for FP64, FP64x2, cascaded arithmetic, respectively.
$ \epsilon_{\rm mach} = 2^{-D} \specific{= 2^{-53}}$, $\widehat \epsilon_{\rm mach} = 2^{-2D}
\specific{= 
2^{-106}}
$,
$\widetilde \epsilon_{\rm mach} = 2^{-(D_2+D)}
\specific{= 
2^{-117}}
$.
\\
\hline
FP64, FP64x2 & Double 
(64 bit) and double-double precision floating point, respectively. \\
\hline
$ \sigma $, $ \tau $
&
Scaling factors for $ x, A $ and $ y, B $, respectively.
\\ 
\hline
$ \Sigma $, $ T $ 
&
Diagonal matrix of scaling factors for $ A $ and $ B $, respectively.
\\
\hline
$ \sigma_1, \sigma_2, \sigma_2 $
&
Scaling factors for bins $ 0, 1, 2 $.
\\
\hline
$ \vert \cdot \vert $, $ \leq $ &
Element-wise absolute value and comparison, respectively, of a vector or matrix.
\\ \hline
\end{tabular}
\end{center}
\caption{Glossary.
Text in \specific{blue} is for the example of cascading FP64x2 into FP64, but can be generalized/modified for other precisions.}
\label{fig:glossary}
\end{table}

\section{Context}

\subsection{Notation}

Since many symbols are used in this paper, we provide a glossary in Table~\ref{fig:glossary}.

\subsection{Basics}
The use of single precision (FP32) and double precision (FP64) storage and arithmetic in \gemm\ is well understood and usually supported in hardware.  A fused-multiply-\response{add} (FMA) in FP64 typically requires twice the time of a FP32 FMA.  More recently, variants on half precision like bfloat16 have become of interest and are often supported in hardware~\cite{bfloat16}.  A precision beyond FP64 is quad precision (FP128), which uses a mantissa of 113 (112 stored) bits and exponent of 15 bits compared to a mantissa of 53 (52 stored) and exponent of 11 bits for FP64.  Hardware support for quad precision is only found on a few processors, largely because the hardware area to support fast FP arithmetic grows quadratically with the mantissa size. Simulating FP128 arithmetic instead in software on the FMA level requires approximately two orders of magnitude more time than a FP64 FMA.

A compromise is double-double (FP64x2) precision, which stores an extended precision number \responsesout{in} \response{as} two FP64 numbers, splitting the mantissa among the two, but providing no greater exponent range. The benefit is that now 106 bits are available for the mantissa, and casting a \response{FP64x2 multiply-add (MAD)} in terms of FP64 FMAs is simpler and requires less computation time~\cite{Bailey88} than FP128 when both are implemented in software.  The problem is that it is still time intensive, with a FP64x2 MAD implemented in software still requiring approximately 20-50 times more time than a hardware FP64 FMA. An early attempt at supporting such extended precision for BLAS functionality was the reference XBLAS~\cite{LiDeBaHeHiIsKaMaThTuYo02,xblasweb}, which make no attempt at optimizing performance.

One problem with early papers on the topic of FP64x2 is that they measure performance by counting flops rather than MADs. Extrapolating from the cost of a FP64 dot product of size $n$ requiring roughly $n$ FP64 multiplies and $n$ FP64 adds, many higher-accuracy dot product strategies would, for example, state an algorithm that requires $20n$ FP64 flops as being $ 10 \times $ more work than FP64 dot products. However, a careful analysis of many of these high precision kernels shows a huge bias toward adds and subtracts over multiplies. This gap is further widened on machines that support FP64 FMAs (nearly all of them), because emulating FP64x2 multiplication only requires 1 FMA and 1 multiply instruction, while FP64x2 addition requires approximately 10 or more add/subtract instructions that are dependent on each other. Therefore, a  FP64x2 multiplication is significantly simpler than FP64x2 addition/subtraction. It is better to have an algorithm with $n$ MADs than one with $1.5n$ adds, even though technically the latter incurs fewer flops.

Another avenue of earlier works focuses on the idea of ``generating a more accurate sum''~\cite{Rump2009} or ``dot-product''~\cite{Ba95,DeHi03,OgRuOi05}.
The idea is that if a dot product algorithm produces a more accurate result, it can be useful in computing a more accurate \gemm. For example, \cite{OzOgOiRu2011} discusses how to split a higher precision \gemm\ via error-free transformations that are not dependent on conditionals, exploiting only regular floating-point add, subtract and multiplication vector instructions. Also, 
 in \cite{HenHeiTan2019} it is discussed how  to cast a FP32 FMA in terms of 
 three bfloat16 FMAs, which is the concept applied in the present work to yield what we call cascading \gemm\ matrices. The same cascading technique was used to show that machine learning could be done strictly with bfloat16 FMAs and could avoid FP32 \gemm, e.g., in  \cite{OsArPeHeCa2022}.

\subsection{Leveraging lower-precision, high-performance \gemm\ implementations}

High-performance FP32 and FP64 implementations of \gemm\ are difficult to implement, in part because careful attention must be paid to overhead caused by moving data between memory layers~\cite{Goto1}.  For this reason, the BLAS~\cite{BLAS3} provide an interface for \gemm\ so that the task of optimizing these can be left to experts.
A question becomes whether a lower-precision, high-performance \gemm\ can be leveraged to implement a higher-precision high-performance \gemm.  

The Ozaki Scheme~\cite{OzakiScheme,OzOgOiRu2011} breaks higher-precision matrices $ A $ and $ B $ up, 
$ A = A_0 + A_1 + \cdots + A_{d_A-1} $ and $ B = B_0 + B_1 + \cdots + B_{d_B-1} $, where each term captures a different part of the range of the elements, allowing these to be stored in lower precision.
Then, in exact arithmetic, $ C = \sum A_i B_j $.  Now each product $ A_i B_j $ becomes a call to a lower-precision \gemm.
These works focus on provable accuracy and reproducibility. \response{ This requires determining the number of splits for the matrices based on the range of the elements of the input matrices and the inner dimension $k$ of the product}. Such range-based quantization methods suffer because as the number of bits needed to  capture the range of the data increases linearly, the number of splits in which to partition increases  linearly, resulting in a quadratic increase in work.

\response{The Ozaki Scheme can be found in OzBLAS~\cite{OzBLAS}, which implements FP128 \gemm\ in terms of FP64.  A ``computational degree'' parameter is introduced, which can be used to limit the number of splits and hence \gemm s, in principle creating tuneable accuracy.%
Possibly inspired by~\cite{HenHeiTan2019},
a ``fast'' mode is introduced where, if $ d_A = d_B = d $,  
a higher-precision \gemm\ is computed in
$\left({d(d+1)}/{2}\right)$ lower-order \gemm s. For example, if $d=d_A=d_B=4$ the fast mode requires ten \gemm s%
\footnote{These ten \gemm s are different than the ten \gemm s that our proposed method will compute as discussed in Section~\ref{sec:CascadingGemm}.}.
}
\response{In addition, as a potential optimization technique, both~\cite{OzBLAS} and~\cite{OzakiScheme} suggest blocking the $m$ and $n$ dimensions of the matrix to reduce the memory required to store the split matrices and~\cite{OzakiScheme} suggests blocking the $k$ dimension to reduce the number of \gemm s that may be required.}

In other closely-related  work that deals with higher accuracy \gemm\ (or dot products) using lower precision computation, the input matrices $A$ and $B$ are stored in the lower precision. For instance,  \cite{MukOzaOgiIma20}
 explores the use of low precision tensor cores on Nvidia GPUs to do cascading matrices with whatever low precision the tensor core supports. Similarly, in~\cite{FaHiLoMaMi22} it is also shown how to perform multi-precision  GEMM.  \cite{OzOgOiRu2011}  already discussed computing $AB$ in higher accuracy, when $A$ and $B$ are in FP64. This case requires about one-half to one-quarter of the work involved in computing the more general case when $A$ and $B$ are in FP64x2. They distinguish between two \response{algorithms}: one that exploits sparsity and skips some GEMM computation, for which they report (for large matrices) $ 13 \times $ slowdown over DGEMM, and one that ignores sparsity that exhibits a  $ 18 \times $ slowdown when $ d_A = d_B  = d = 4 $.   When accuracy requires more splits, the number of required \gemm s can easily exceed 150, for the simpler problem where $ A $ and $ B $ start as FP64 matrices.

In~\cite{FaHiLoMaMi22}, the work on breaking up higher precision \gemm s in terms of lower precision has been extended.  Fasi et al.~discuss the error analysis associated with cascading FP32 matrices into bfloat16 matrices to use Nvidia's tensor core to obtain FP32 precision \gemm. Their proposed technique and error analysis rely on the fact that while the input to the tensor core is bfloat16, the internal accumulator employs FP32 precision. 

 \subsection{Our work}

We primarily consider the case study where $A$, $ B $, and $ C $ are FP64x2 matrices and the goal is to approximate computing $ C = A B $ with   \response{FP64x2 precision}, without the benefit of hardware support for FP64x2 accumulation.    We pursue high accuracy balanced  with high performance, allowing only a fixed number of calls to FP64 \gemm s so that time to completion is predictable (not data dependent).  In the implementation, we fuse various required  phases in order to reduce overhead due to memory movement of data, allowing high performance  to be attained already for relatively small problem sizes.

Constrained by these goals, we strive to maximize the attained accuracy while minimizing the number of required  \gemm\ calls, 
 making the amount of work done independent of the input data range and the inner $k$ dimension of the matrix. This compromise facilitates a high performance implementation that allows us to approximate FP64x2 GEMM in $10$-$13 \times$ DGEMM time.  Key is the observation that multiple ``bins,'' used to accumulate lower precision contributions, can be combined.  In comparison, by counting up the flops in the XBLAS, one could imagine DDGEMM to run in $20$-$50 \times$ DGEMM time, but we are unfamiliar with such an optimized implementation. Moreover, we show how the cascading matrix multiplication can leverage principles that underlie BLIS to limit the number of splits and reduce overhead from data movement between memory layers.  
 In all, in our opinion, these represent a significant new contribution to the subject.

\section{Cascading}

We illustrate the fundamental ideas with a case study that shows how to cascade FP64x2 \gemm\ into FP64 \gemm s.  How the techniques extend will be discussed later.

\subsection{FP64 and FP64x2 numbers}
\label{sec:FP64x2}

A normalized FP64 number, $  \chi $, consists of a sign, a mantissa, and an exponent:
\[
\chi = 
\pm \period \beta_0 \cdots  \beta_{\digits-1} \times 2^e
=
\pm
\left(
\beta_0 \times 2^{-1} 
+
\beta_1 \times 2^{-2} 
+
\cdots
+
\beta_{D-1} \times 2^{-D} 
\right)
\times2^e,
\]
where $ \beta_0 = 1 $.
A typical $ D $ equals 53 and for that reason that integer will frequently appear in our discussions.

A double-double floating point number (FP64x2),
$ \widehat \chi $,
increases the size of the mantissa by using two normalized FP64 numbers:
\begin{equation}
\label{eqn:doubledouble}
\widehat \chi 
= 
\chi_{\rm hi} + 
\chi_{\rm lo} \times 2^{-D}.
=
\begin{array}[t]{c}
\underbrace{
\pm \period \beta_0 \cdots  \beta_{\digits-1} \times 2^e
}
\\
\chi_{\rm hi}
\end{array}
+
\begin{array}[t]{c}
\underbrace{
\pm \period \beta_{\digits+\gap} \cdots  \beta_{2\digits+\gap-1} \times 2^{e-\gap} 
}
\\
\chi_{\rm lo}
\end{array}
\times 2^{-D}.
\end{equation}
Here,
\begin{itemize}
\item
We add a $ \widehat{~} $ to denote a FP64x2 variable.
\item
There are no overlapped (ranges of) bits in $  \chi_{\rm hi} $ and 
$ \chi_{\rm lo} $,
\item
$ \gap \geq 0 $ captures \response{the scenario} that there could be leading zeroes 
(a {\em gap}) right where the split of the FP64x2 number happens.
\end{itemize}
\response{If one is lucky, this gap allows more accuracy to be stored (in the place where zero bits occur) in a FP64x2 than a $ 2D $ bit mantissa would. }
However, one can obviously not count on that extra accuracy and hence in our discussion we will mostly assume $ \gap = 0 $ and will refer to it as a {\em lucky gap} to emphasize that it cannot be counted on.

\subsection{Cascading scalars and scalar multiplication}
\label{sec:cascading_scalar}

\remove{
We start by examining how to one might think of a FP64x2 in terms of two FP64s, which we will refine subsequently.
}

Consider the FP64x2 number discussed in Section~\ref{sec:FP64x2}, $ \widehat \chi = {\chi_{\rm hi}} + {\chi_{\rm lo}} \times 2^{-D}$ and a second FP64x2 number $ \widehat \psi = \psi_{\rm hi} + \psi_{\rm lo} \times 2^{-D}$.
\remove{By normalized, we mean no overlapping bits between the doubles $\hat{\chi_0}$ and $\hat{\chi_1}$. The exponent of the low part $\hat{\chi_1}$ is then at least 53 bits from the exponent of the high part $\hat{\chi_0}$, but it could be more.
The bits of $\widehat \chi$ can be written as:
\[
\chi = \pm \period \beta_0 \cdots  \beta_{\digits-1} \beta_{\digits} \cdots \beta_{2\digits-1} \times 2^e,
\]
where $ \digits $ equals 
the number of bits stored in the mantissa of a double precision float (FP64)%
\footnote{While this is not how FP64x2 is physically stored, it simplifies our discussion.}.  
A typical $ \digits $ equals $ 53 $.
 We can represent $ \chi $ exactly as two FP64s that share the same exponent:
\[
\chi = \chi_0  + \chi_1 \times 2^{-\digits},
\]
where $ \chi_0 = \pm  \period  \beta_0 \beta_1 \cdots \beta_{\digits-1} \times 2^e $ and $ \chi_1 = \pm \period \beta_{\digits} \cdots \beta_{2\digits-1} \times 2^e  $.
Note that $\hat{\chi_0}=\chi_0$ and $\hat{\chi_1}=\chi_1 \times 2^{-D}.$
Notice that $\chi_0$ and $\chi_1$ now have the same exponent. A number of the most significant bits of $\chi_1$ may be zero, which only means that $D>53$, but again, this does not impact our discussion.
}\remove{
If we then do the same for a second FP64x2, $ \psi = \psi_0  + \psi_1 \times 2^{-\digits} $,
and multiply, we get
}
Then
\[
\widehat \chi  \widehat \psi = \chi_{\rm hi}  \psi_{\rm hi} + \chi_{\rm hi} \psi_{\rm lo} \times 2^{-D} + \chi_{\rm lo}  \psi_{\rm hi} \times 2^{-D} + \chi_{\rm lo}  \psi_{\rm lo} \times 2^{-2D}.
\]
This seems to illustrate how a single FP64x2 multiplication can be broken down into four FP64 multiplications, some shifts of  the exponents, and three FP64x2 additions. 
The problem is that each of the four terms may incur error when computed in FP64,  including in the product of the highest order term, $ \chi_0 \psi_0 $, 
because a FP64 multiplication can produce $ 2 D $ bits in the mantissa of the result which are stored back into a $ D $ bit FP64 before adding the terms together.

One can
try to overcome
these multiplication inaccuracies by approximating
\[
\widehat \chi \approx \chi_0 + \chi_1 \times 2^{-D_0},
\]
where 
$ \chi_0 = \pm  \period  \beta_0 \beta_1 \cdots \beta_{\digits_0-1} \times 2^e $ and $ \chi_1 \approx \pm \period \beta_{\digits_0} \cdots \beta_{\digits_0 + \digits-1} \times 2^e  $.
Now, if $ D_0 \le D/2 $ then
$ \chi_0 \psi_0 $ is computed exactly in FP64 arithmetic since the result has at most $ D $ bits in the mantissa.  However, one would expect a loss in accuracy since digits beyond $ \beta_{D_0 + D -1} $ vanish in the conversion (unless there is a {\em lucky gap}).

With this introduction, we instead 
consider cascading (splitting) a scaled FP64x2, $ \widehat \chi $, 
into four FP64s (splits).
We start with
\[
\widehat \chi / \sigma = \pm  \period  \beta_0 \cdots  \beta_{ \digits_0 -1} \beta_{ \digits_0 } \cdots \beta_{\digits_1-1}
\beta_{ \digits_1} \cdots \beta_{ \digits_2-1}
\beta_{ \digits_2} \cdots \beta_{\digits_2 + \digits-1} 
\beta_{\digits_2 + \digits}
\cdots,
\]
where $ \sigma $ is a power of two 
\response{and there may be leading zeroes.}
In other words, $ \sigma \leq 2^e $, where $ e $ is the exponent that occurs in (\ref{eqn:doubledouble}).
We can now ``cascade'' this scaled FP64x2 number into four FP64s:
\[
\widehat \chi / \sigma \approx \pm \period
\begin{array}{|l|l|l|l|} \hline
\beta_0 \cdots  \beta_{ \digits_0 -1} & \beta_{ \digits_0 } \cdots \beta_{\digits_1-1} &
\beta_{\digits_1} \cdots \beta_{\digits_2-1} &
\beta_{\digits_2} \cdots \beta_{\digits_2 + \digits  -1}
\\
\hline
\end{array}
,
\]
and represent (approximately) $ \widehat \chi / \sigma  $ in terms of four FP64s:
\[
\begin{array}{rrlllll}
\widehat \chi / \sigma &\approx& 
\begin{array}[t]{c}
\underbrace{
\pm
\period \beta_0 \cdots  \beta_{ \digits_0 -1} 
} \\
\chi_0
\end{array}
&&+
\begin{array}[t]{c}
\underbrace{
\pm
\period \beta_{ \digits_0 } \cdots \beta_{  \digits_1-1} 
 }
 \\
 \chi_1
 \end{array}
 & \times 
 \begin{array}[t]{c}
\underbrace{
2^{ -\digits_0 }}
\\
\sigma_1
\end{array}
\\
&+&
\begin{array}[t]{c}
\underbrace{
\pm
\period \beta_{\digits_1} \cdots \beta_{ \digits_2-1}
}
\\
\chi_2
\end{array}
& \times 
\begin{array}[t]{c}
\underbrace{
2^{ -\digits_1}}
\\
\sigma_2
\end{array}
&+
\begin{array}[t]{c}
\underbrace{
\pm
\period
\beta_{ \digits_2} \cdots \beta_{\digits_2 + \digits-1} 
}
\\
\chi_3
\end{array}
& \times 
\begin{array}[t]{c}
\underbrace{
2^{ -\digits_2}}
\\
\sigma_3
\end{array}
.
\end{array}
\]
This is also illustrated in Figure~\ref{fig:chunks}. 
\remove{Of course, there may be leading zero bits in each chunk, but that will be taken care of automatically within the FP64 number.
}
More concisely
\[
\widehat \chi \approx
\sigma
\left(
\chi_0 + \sigma_1 \chi_1  
+ \sigma_2 \chi_2  + \sigma_3 \chi_3  \right),
\]
where $\chi_0, \chi_1$ and $\chi_2$ have a fixed bit-width of $D_0, D_1-D_0$ and $D_2-D_1$ respectively and $\chi_3$ captures the error of the first 3 terms in as much precision as FP64 can allow (up to 53 bits).
If $ \beta_0 = 1 $, $ D_2 + D > 2D $, and there is no lucky gap in the FP64x2 representation, then the cascaded representation can store more digits than FP64x2.
Even more precision may result if there there is a gap at the beginning of $ \chi_3 $, but that cannot be counted on.
In other words, 
if the cascaded number holds the result of other computations, it may hold extra precision in the trailing zeroes of $ \chi_3 $.
We consider the existence of a cascading split for any FP64x2 number to be self-evident.

\begin{figure}[tbp]
\begin{center}
{
\setlength{\unitlength}{0.35in}
\begin{picture}(17,3.2)
\put(0.4,2.8){\makebox(0,0){$\widehat \chi/\sigma = $}}
\put(1.2,2.8){\makebox(0,0){$\pm \period $}}
\put(15.7,2.8){\makebox(0,0){$ \phantom{\times 2^{-D_0}} $}}
\put(1.5,2.4){
\put(0,0.8){\line(1,0){5.8}}
\put(0,0){\line(1,0){5.8}}
\put(0,0){\line(0,1){0.8}}
\put(2.4,0){\line(0,1){0.8}}
\put(5.8,0){\line(0,1){0.8}}
\put(1.2,0.4){\makebox(0,0){$\beta_0 \cdots \beta_{D_0-1}$}}
\put(4.1,0.4){\makebox(0,0){$0 ~~~~~~~~ \cdots ~~~~~~~~ 0 $}}
}

\put(1.2,2.0){\makebox(0,0){$\pm \period $}}
\put(15.2,2.0){\makebox(0,0){$ \times 2^{-D_0} $}}
\put(3.9,1.6){
\put(0,0.8){\line(1,0){5.8}}
\put(0,0){\line(1,0){5.8}}
\put(0,0){\line(0,1){0.8}}
\put(2.4,0){\line(0,1){0.8}}
\put(5.8,0){\line(0,1){0.8}}
\put(1.2,0.4){\makebox(0,0){$\beta_{D_0} \cdots \beta_{D_1-1}$}}
\put(4.1,0.4){\makebox(0,0){$0 ~~~~~~~~ \cdots ~~~~~~~~ 0 $}}
}

\put(1.2,1.2){\makebox(0,0){$\pm \period $}}
\put(15.2,1.2){\makebox(0,0){$ \times 2^{-D_1} $}}
\put(6.3,0.8){
\put(0,0.8){\line(1,0){5.8}}
\put(0,0){\line(1,0){5.8}}
\put(0,0){\line(0,1){0.8}}
\put(2.4,0){\line(0,1){0.8}}
\put(5.8,0){\line(0,1){0.8}}
\put(1.2,0.4){\makebox(0,0){$\beta_{D_1} \cdots \beta_{D_2-1}$}}
\put(4.1,0.4){\makebox(0,0){$0 ~~~~~~~~ \cdots ~~~~~~~~ 0 $}}
}

\put(1.2,0.4){\makebox(0,0){$\pm \period $}}
\put(15.2,0.4){\makebox(0,0){$  \times 2^{-D_2} $}}
\put(8.7,0.0){
\put(0,0.8){\line(1,0){5.9}}
\put(0,0){\line(1,0){5.9}}
\put(0,0){\line(0,1){0.8}}
\put(5.9,0){\line(0,1){0.8}}
\put(2.9,0.4){\makebox(0,0){$\beta_{D_2} ~~~~~~ \cdots ~~~~~~ \beta_{2D-1} ~~0 ~ \cdots ~ 0 $}}
}
\end{picture}
}
\end{center}
\caption{Illustration of how FP64x2 number $ \widehat \chi $ is cascaded into four FP64 splits.  Here we assume that we start with a normalized number so that $ \beta_0 = 1$.  If $ D_i \leq D $ for $ i = 0, 1, 2 $, then this also illustrates that additional precision can be stored in the cascaded number, which can improve accuracy when the cascaded representation is used for intermediary accumulation.}
\label{fig:chunks}
\end{figure}
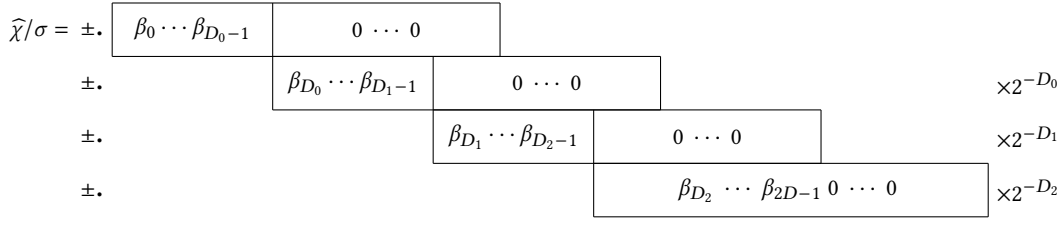

If we similarly cascade a second FP64x2 number, 
$ \widehat \psi $:
\[
\widehat \psi \approx
\tau \left( 
\psi_0 + 
\sigma_1 \psi_1  + 
\sigma_2 \psi_2  +
\sigma_3 \psi_3 \right)  ,  \]
then the product of these two scalars equals, approximately,
\[
\begin{array}{l@{}ll@{~}l@{~}@{~}r@{~}r@{~}r@{~}r@{~}r@{~}r@{~}r@{~}r@{~}r@{~}r@{~}r@{~}r @{~}r}
\widehat \chi 
\widehat \psi \approx \sigma \tau ( ~~~& 
\chi_0 \psi_0 &+&
\sigma_1 \chi_0 \psi_1   &+&
\sigma_2 \chi_0 \psi_2    &+&
\sigma_3 \chi_0 \psi_3    \\
&&+&
\sigma_1 \chi_1 \psi_0  &+&
\sigma_1^2 \chi_1 \psi_1  &+&
\sigma_1 \sigma_2 \chi_1 \psi_2   &+&
\sigma_1 \sigma_3 \chi_1 \psi_3  &+& \\
&&&&+&
\sigma_2 \chi_2 \psi_0   &+&
\sigma_2 \sigma_1 \chi_2 \psi_1  &+&
\sigma_2^2 \psi_2   &+&
\sigma_2 \sigma_3 \chi_2 \psi_3  &+& \\
&&&&&&+&
\sigma_3  \chi_3 \psi_0  &+&
\sigma_3 \sigma_1 \chi_3 \psi_1  &+&
\sigma_3 \sigma_2 \chi_3 \psi_2   &+&
\sigma_3^2 \chi_3 \psi_3  & )  . \\[-0.1in]
&
\begin{array}[t]{@{}c@{}}
\underbrace{
\phantom{\chi_0 \psi_0 }
}
\\
\mbox{bin 0}
\end{array}
&&
\begin{array}[t]{@{}c@{}}
\underbrace{
\phantom{\chi_0 \psi_0 \sigma_0 }
}
\\
\mbox{bin 1}
\end{array}
&&
\begin{array}[t]{@{}c@{}}
\underbrace{
\phantom{\chi_0 \psi_0 \sigma_0 }
}
\\
\mbox{bin 2}
\end{array}
&&
\begin{array}[t]{@{}c@{}}
\underbrace{
\phantom{\chi_0 \psi_0 \sigma_0 \tau_0}
}
\\
\mbox{bin 3}
\end{array}
&&
\begin{array}[t]{@{}c@{}}
\underbrace{
\phantom{\chi_0 \psi_0 \sigma_0 \tau_0}
}
\\
\mbox{bin 4}
\end{array}
&&
\begin{array}[t]{@{}c@{}}
\underbrace{
\phantom{\chi_0 \psi_0 \sigma_0 \tau_0}
}
\\
\mbox{bin 5}
\end{array}
&&
\begin{array}[t]{@{}c@{}}
\underbrace{
\phantom{\chi_0 \psi_0 \sigma_0 \tau_0}
}
\\
\mbox{bin 6}
\end{array}
\end{array}
\]
Notice that 
\[
\begin{array}{lrcl @{\hspace{1in}} lrcl}
\mbox{bin 1:} &
 \sigma_1  &=& 2^{-D_0} &
\mbox{bin 4:} &
\sigma_2^2 \approx \sigma_1 \sigma_3   & \ll & 2^{-D_2} \\
\mbox{bin 2:} &
\sigma_1^2 \approx  \sigma_2   & = & 2^{- D_1} &
\mbox{bin 5:} &
\sigma_2 \sigma_3  & \ll & 2^{-D_2} \\
\mbox{bin 3:} &
\sigma_1 \sigma_2   \approx \sigma_3  & = & 2^{- D_2} &
\mbox{bin 6:} &
\sigma_3^2 & \ll & 2^{-D_2}.
\end{array}
\]

\NoShow{\devangi{\emph{In the highlighted equation above, are we off by one, since $D_0 =22$ but $D_1 = 21$? or does the $\approx$ take care of that? Oh wait, it does, right? the plus or minus one for some reason I interpreted as the sign.}}}
The ``approximate'' here means that they are equal plus or minus one, if the sizes of splits are chosen appropriately.  
If $ D_0 $, $ D_1 $, and $ D_2 $ are chosen carefully%
\footnote{How to choose these parameters will become clear shortly.}, the first three bins can be computed and stored exactly in FP64.

Thus, the process becomes:
\begin{enumerate}
    \item 
    Transform $ \widehat \chi $ and $ \widehat \psi $ into cascaded representation.
    \item
    Compute within the bins  in FP64 arithmetic.
    \item
    Add within and across the bins in appropriate precision.
\end{enumerate}
We note that accumulating the contributions across a few of the higher-order bins (bin~0, bin~1 and bin~2.) must be in higher than double precision (in FP64x2 arithmetic or some compensated two-sum algorithm~\cite{Dekker1971}) and the rest can be summed in FP64. 

Let us analyze the error introduced when cascading a real number: 
\[
\widehat \chi = \sigma
\left(
\chi_0 + \chi_1 \sigma_1 
+ \chi_2 \sigma_2 + \chi_3 \sigma_3 + \delta\!\chi \right) ,
\]
where $ \vert \delta\!\chi \vert \leq 2^{-(D_2 + D)} $ if rounding is used to create the fourth split.
In $ \delta\!\chi $ the $ \delta $ touches the $ \chi $ to indicate this is one symbol to represent the error term following the conventions in~\cite{Bientinesi:2011:GMS:2078718.2078728}.
This means that 
\[
\widehat \chi =
\sigma \left(
\chi_0 + \chi_1 \sigma_1 
+ \chi_2 \sigma_2 + \chi_3 \sigma_3 \right)  + \sigma \delta\!\chi,
\]
where 
$ \vert \sigma \delta\!\chi  \vert \leq 2^{-(D_2 + D)} \sigma $.
This bound is pessimistic, since the last split, $ \chi_3 $ will hold a full FP64 number and hence even if $ \chi_0 = \chi_1 = \chi_2=0 $, 
the last split holds a FP64 appproximation of the FP64x2 number. This allows us to refine the bound to
\begin{equation}
\label{eqn:error_scalar}
 \vert \delta\!\chi  \sigma \vert \leq \min( 2^{-(D_2 + D)} \sigma,
\macheps \widehat \chi ) = 2^{-(D_2 + D)} \sigma,
\end{equation}
where $ \macheps = 2^{-D} $ and \response{ $ 2^{-(D_2 + D )}$ are the {\em machine epsilon} for an FP64 and a  cascaded number, respectively. Going forward, we will refer to the machine epsilon of the cascading number as $\machepsnew$}.

\subsection{Cascading vectors and dot products}
\label{sec:cascading_dot}

On our way to understanding the properties of using cascading arithmetic within a \gemm, we propose how to cascade vectors and examine a dot product with such vectors.

Let $ \widehat x $ and $ \widehat y $ be vectors of size $ k $ with FP64x2 numbers as their entries.  We wish to compute
$
\widehat x^T \widehat y
$.
We cascade each of the vectors:
\[
\begin{array}{rcl@{~}c@{~}l@{~}c@{~}l@{~}c@{~}l@{~}c@{~}l@{~}c@{~}l}
\widehat x & \approx & \sigma~( & x_0 &+&
\sigma_1 x_1  &+&
\sigma_2 x_2  &+&
\sigma_3 x_3 & )& ~ \\
\widehat y & \approx & \tau ~( & y_0  & + &
\sigma_1 y_1  &+&
\sigma_2 y_2  &+&
\sigma_3 y_3  &)
\end{array}
\]
so that 
\begin{equation}
    \label{eq:dot}
\begin{array}{ll@{\hspace{2pt}}c@{\hspace{2pt}}@{\hspace{2pt}}r@{\hspace{2pt}}r@{\hspace{2pt}}r@{\hspace{2pt}}r@{\hspace{2pt}}r@{\hspace{2pt}}r@{\hspace{2pt}}r@{\hspace{2pt}}r@{\hspace{2pt}}r@{\hspace{2pt}}r@{\hspace{2pt}}r@{\hspace{2pt}}l}
\widehat x^T \widehat y \approx \sigma \tau ~ ( & 
 x_0^T y_0  &+&
\sigma_1 x_0^T y_1  &+&
\sigma_2 x_0^T y_2   &+&
\sigma_3 x_0^T y_3    \\
&&+&
\sigma_1 x_1^T y_0   &+&
\sigma_1^2 x_1^T y_1  &+&
\sigma_1 \sigma_2 x_1^T y_2   &+&
\sigma_1 \sigma_3 x_1^T y_3  &+& \\
&&&&+&
\sigma_2  x_2^T y_0 
&+&
\sigma_2 \sigma_1 x_2^T y_1  &+&
\sigma_2^2 x_2^T y_2   &+&
\sigma_2 \sigma_3 x_2^T y_3  &+& \\
&&&&&&+&
\sigma_3 x_3^T y_0  
&+&
\sigma_3 \sigma_1 x_3^T y_1  &+&
\sigma_3 \sigma_2 x_3^T y_2   &+&
\sigma_3^2  x_3^T y_3   & ) .  \\\\[-0.1in]
&
\begin{array}[t]{@{}c@{}}
\underbrace{
\phantom{\chi_0 \psi_0 }
}
\\
\mbox{bin 0}
\end{array}
&&
\begin{array}[t]{@{}c@{}}
\underbrace{
\phantom{\chi_0 \psi_0 \sigma_0 }
}
\\
\mbox{bin 1}
\end{array}
&&
\begin{array}[t]{@{}c@{}}
\underbrace{
\phantom{\chi_0 \psi_0 \sigma_0 \tau_0}
}
\\
\mbox{bin 2}
\end{array}
&&
\begin{array}[t]{@{}c@{}}
\underbrace{
\phantom{\chi_0 \psi_0 \sigma_0 \tau_0}
}
\\
\mbox{bin 3}
\end{array}
&&
\begin{array}[t]{@{}c@{}}
\underbrace{
\phantom{\chi_0 \psi_0 \sigma_0 \tau_0}
}
\\
\mbox{bin 4}
\end{array}
&&
\begin{array}[t]{@{}c@{}}
\underbrace{
\phantom{\chi_0 \psi_0 \sigma_0 \tau_0}
}
\\
\mbox{bin 5}
\end{array}
&&
\begin{array}[t]{@{}c@{}}
\underbrace{
\phantom{\chi_0 \psi_0 \sigma_0 }
}
\\
\mbox{bin 6}
\end{array}
\end{array}
\end{equation}
We now examine under what circumstances splits $ i,j \in \{ 0,1,2\} $ of 
$ x_i^T y_j $ are guaranteed to be computed exactly. Consider that
\[
x_i^T y_j 
= \chi_{0,i} \psi_{0,j} + \chi_{1,i} \psi_{1,j} + \cdots + \chi_{k-1,i} \psi_{k-1,j},
\]
where $\chi_{p,i}$ is the $p$-th element of the cascaded vector $x_i$.
In order to compute each $ \chi_{p,i} \psi_{p,j} $ exactly {\em and} add these results together to be exact, we observe that:
\begin{itemize}
    \item 
    Each element in $ x $ must be split into the same ranges.  This means
    \begin{itemize}
        \item 
        Finding the element with largest magnitude and \response{dividing} all terms in that vector by the smallest power of two greater than that value (normalizing everything in the vector.)
        This tells us how to choose $ \sigma $ and $ \tau $.
        \item
        Splitting all other elements conformally, by which we mean that corresponding splits are taken from the same range.
    \end{itemize}
    Observe that elements other than the one with \response{the} largest magnitude may have leading zeroes, including in their most significant split%
    \footnote{While some of this can be addressed within the FP64 number, since it is a floating point number, there is the possibility of losing precision in the conversion.}.
    \item
    Adding $ k $ contributions from the terms in the dot product may yield a result with $ \lceil \log_2 k \rceil $ additional binary digits that must be tracked.
    Adding the products within a bin may magnify the number of bits further.
\end{itemize}
Assume that $ x $ and $ y $ are split \response{ such that each split accommodates $c_i$ bits, where $i\in \{ 0,1,2,3 \}$. Here,} $c_0 = D_0$, $ c_1 = D_{1} - D_{0} $, $c_2 = D_2 - D_1$  and $ c_3 \le 2D - D_2 $, and \response{we assume that} $ k = 256 $ so that $ \log_2 k = 8 $ as an example%
\footnote{
We will see that choosing $ k $ to be relative small
naturally occurs in high-performance algorithms for \gemm.
If $ k $ is smaller or larger, appropriate adjustments can be made.}. 
Then
\begin{itemize}
\item
To ensure that 
bin 0 is computed in full accuracy, we constrain $ c_0 $ such that  $ 2 c_0  + \log_2 k \leq 53 $.
This suggests that $ c_0 \leq (53 - 8)/2 $.
Thus, split 0 can accommodate $ c_0 = 22 $ bits. \NoShow{\greg{as long as we continue to block K to be no bigger than 256, and even more bits otherwise.}}
\NoShow{\remove{In our BLIS runs on Intel(R) Xeon processors, $k$ is blocked by 256 by DGEMM (hence our choice of $k=256$ above.) But notice that if $k$ happens to be smaller than 256, we can pick a slightly larger $c_0$ in our implementation. But for larger problems, $c_0=22$ is a fair choice.}
}
\item
To ensure that 
bins 1 and 2 are computed in full accuracy, we must constrain $ c_1 $ such that  both $ c_0 + c_1 + \log_2 k + \lceil \log_2 2 \rceil \leq 53 $
and 
$ 2c_1 + \log_2 k + \lceil \log_2 3 \rceil \leq 53 $.  The $ 2 $ and $ 3 $ come from  the number of terms that need to be added in bin 1 and bin 2 respectively.
This implies that split 1 can accommodate only $ c_1 = 21 $ bits.
\item
To ensure that 
bin 2 is computed in full accuracy, we constrain $ c_2 $ such that   $ c_0 + c_2 + \log_2 k + \lceil \log_2 4 \rceil \leq 53 $, which implies that split 2 can accommodate $ c_2 = 21 $ bits. 
\NoShow{
\remove{And again, if $k$ happens to be smaller, we can be more aggressive here and perhaps try $c_2=22.$}}
\item
The remaining bins (3-6) contribute to lowest order terms and hence they place no constraints on the splittings.
In other words, $c_3=53$. 
\end{itemize}
The total number of bits stored with cascaded \response{representation} is now at least $ 22+21+21+53 = 117 $ which is more than the number of bits in the mantissa of a FP64x2 number (if there is no lucky gap). Note that when $k<256$, we can store even more than 117 bits. That is, in our implementation, when choosing the sizes of $c_i$, we actually look at $k$ as well, and if $k=64$, for instance, then we allow $c_0=23$. Regardless of the size of $k$, the minimal 117 bits means intermediate results are potentially accumulated in a precision higher than FP64x2 accommodates.
Summation within bins can be in FP64 arithmetic.
Summation across bins 0 through 2 must be performed in FP64x2 addition or with quick two-sum arithmetic~\cite{Dekker1971}. Summation across bins~3 through~6 can be done in FP64 arithmetic, since $c_3=53$ implies that these terms are not even attempted to be done error-free.  
To further preserve accuracy, the
 \response{addition} of contributions across bins starts with bin~6 and ends with bin~0.

The computation in~(\ref{eq:dot}) can be reformulated as the computation of the sixteen  terms
\[
\left( \begin{array}{r}
x_0^T \\ \hline
\sigma_1 x_1^T \\ \hline
\sigma_2 x_2^T \\ \hline
\sigma_3 x_3^T 
\end{array} \right)
\begin{array}{c}
\left( \begin{array}[t]{c | c | c | c}
y_0 & \sigma_1 y_1 & \sigma_2 y_2 & \sigma_3 y_3 
\end{array} \right) \\
~~ \\
~~ \\
~~
\end{array}
=
\left( \begin{array}{r | r | r | r}
 x_0^T y_0 &  \sigma_1 x_0^T y_1 &  \sigma_2 x_0^T y_2 & \sigma_3 x_0^T y_3 \\ \hline
\sigma_1  x_1^T y_0 & \sigma_1^2 x_1^T y_1 & \sigma_1 \sigma_2 x_1^T y_2 & \sigma_1 \sigma_3 x_1^T y_3 \\ \hline
\sigma_2  x_2^T y_0 & \sigma_2 \sigma_1 x_2^T y_1 & \sigma_2^2 x_2^T y_2 & \sigma_2 \sigma_3 x_2^T y_3 \\ \hline
\sigma_3  x_3^T y_0 & \sigma_3 \sigma_1 x_3^T y_1 & \sigma_3 \sigma_2 x_3^T y_2 & \sigma_3^2 x_3^T y_3 
\end{array}
\right)
\]
followed by the addition of these terms, yielding a single scalar.
\NoShow{\remove{This will become important in our discussion of the high performance implementation of cascading matrices in Section~\ref{sec:BLISCascadingGemm}.}}

Finally, we observe an alternative for computing  the sum of bins 3--6: 
\begin{equation}
    \label{eqn:xy}
\begin{array}{r@{\hspace{2pt}}r@{\hspace{2pt}}r@{\hspace{2pt}}r@{\hspace{2pt}}r@{\hspace{2pt}}r@{\hspace{2pt}}r@{\hspace{2pt}}r@{\hspace{2pt}}r@{\hspace{2pt}}r@{\hspace{2pt}}r@{\hspace{2pt}}r@{\hspace{2pt}}r@{\hspace{2pt}}r@{\hspace{2pt}}r}
\mbox{bins 3--6} =  &
\sigma_3 x_0^T y_3    
&+&\\
& \sigma_1 \sigma_2 x_1^T y_2   &+&
\sigma_1 \sigma_3  x_1^T y_3 &+& \\
& \sigma_2 \sigma_1 x_2^T y_1  &+&
\sigma_2^2 x_2^T y_2   &+&
\sigma_2 \sigma_3 x_2^T y_3  &+& \\
&
\sigma_3 x_3^T y_0  &+&
\sigma_3 \sigma_1 x_3^T y_1 &+&
\sigma_3 \sigma_2 x_3^T y_2  &+&
\sigma_3^2  x_3^T y_3 \\[-0.1in]
&
\begin{array}[t]{@{}c@{}}
\underbrace{
\phantom{\chi_0 \psi_0 \sigma_0 \tau_0}
}
\\
\mbox{bin 3}
\end{array}
&&
\begin{array}[t]{@{}c@{}}
\underbrace{
\phantom{\chi_0 \psi_0 \phantom{\sigma_0 \tau_0}}
}
\\
\mbox{bin 4}
\end{array}
&&
\begin{array}[t]{@{}c@{}}
\underbrace{
\phantom{\chi_0 \psi_0 \phantom{\sigma_0 \tau_0}}
}
\\
\mbox{bin 5}
\end{array}
&&
\begin{array}[t]{@{}c@{}}
\underbrace{
\phantom{\chi_0 \psi_0 \sigma_0 \tau_0}
}
\\
\mbox{bin 6}
\end{array}
\\
= &
\multicolumn{9}{l}{
\sigma_3 x_0^T  y_3   
+
\sigma_1 x_1^T 
\begin{array}[t]{c}
\underbrace{
\left( \sigma_2  y_2  +
\sigma_3 y_3 \right)
} \\
\sigma_2 y_4 
\end{array} 
+
\sigma_2 x_2^T 
\begin{array}[t]{c}
\underbrace{
\left( \sigma_1 y_1   +
\sigma_2 y_2   +
\sigma_3 y_3  
\right) 
} \\
\sigma_1 y_5 
\end{array}
}
\\
&
\multicolumn{8}{l}{
~~~~~~~~~ +
\sigma_3 x_3^T 
\begin{array}[t]{c}
\underbrace{
\left( y_0   +
\sigma_1 y_1  +
\sigma_2 y_2 
+
\sigma_3 y_3 
\right)
} \\
y_6 
\end{array}
}
\\
=
&
\multicolumn{7}{l}{
\sigma_3 x_0^T y_3  
+
\sigma_1 \sigma_2 x_1^T y_4 
+
\sigma_2 \sigma_1 x_2^T y_5 
+
\sigma_3 x_3^T y_6,
}
\end{array}
\end{equation}
where $y_4 = y_2 + \left({\sigma_3}/{\sigma_2}\right)y_3$, $y_5 = y_1 + \left({\sigma_2}/{\sigma_1}\right)y_2 + \left({\sigma_3}/{\sigma_1}\right)y_3$ and 
$y_6 = y_0 + {\sigma_1}y_1 + {\sigma_2}y_2 + {\sigma_3}y_3.$ 
\response{This is {\bf not} the same as dropping the products from bins 4 to 6 altogether, even though both techniques will lead to 10 products.}

This also shows how the approximate FP64x2 dot product can be cascaded into {\bf ten} rather than sixteen FP64 dot products. Viewing splits of $ x $ as rows of a $ 4 \times k $ matrix and splits of $ y $ and related vectors as columns in a $ k \times 7 $ matrix, 
(\ref{eq:dot})  can be formulated
as 
\begin{eqnarray*}
\lefteqn{
\left( \begin{array}{c}
\phantom{\sigma_0} x_0^T \\ \hline
\sigma_1 x_1^T \\ \hline
\sigma_2 x_2^T \\ \hline
\sigma_3 x_3^T 
\end{array} \right)
\begin{array}{c}
\left( \begin{array}[t]{c | c | c | c || c | c | c }
 y_0 & \sigma_1 y_1 & \sigma_2 y_2 & \sigma_3 y_3  & \sigma_2 y_4  & \sigma_1 y_5  & \tau_0 y_6 
\end{array} \right) \\
~~ \\
~~ \\
~~
\end{array}
= } \\
& \left( \begin{array}{c | c | c | c || c | c | c }
\phantom{\sigma_0 \tau_0} x_0^T y_0 & \phantom{\sigma_0} \sigma_1 x_0^T y_1 & \phantom{\sigma_0} \sigma_2 x_0^T y_2 & \phantom{\sigma_0} \sigma_3 x_0^T y_3 
 & \star & \star  & \star \\ \hline
\sigma_1  x_1^T y_0 & \sigma_1 \sigma_1 x_1^T y_1 & \star & \star &
\sigma_1 \sigma_2 x_1^T y_4 & \star & \star \\ \hline
\sigma_2  x_2^T y_0 & \star & \star  & \star &
\star & \sigma_2 \sigma_1 x_2^T y_5 & \star \\ \hline
\star & \star & \star & \star &
\star & \star & \sigma_3  x_3^T y_6
\end{array}
\right)
\end{eqnarray*}
followed by the summation of the elements or the appropriate contributions into bins and then the accumulation of those, in appropriate precision.   Here the $\star$ entries do not need to be computed.

We note that
\[
\widehat x =
( x_0 + x_1 \sigma_1 + x_2 \sigma_2 + x_3 \sigma_3 + \delta\!x ) \sigma 
,
\mbox{~where~} \vert ~ \delta\!x  ~ \vert \leq \machepsnew \vec \jmath .
\]
In the symbol $ \delta\!x $ the $ \delta $ touches the $ x $ to indicate this is one symbol to represent the error vector;
$ \vert ~ \cdot ~ \vert $ returns the element-wise absolute value; and 
$ \leq $ is an element-wise comparison; and   $ \vec \jmath $ is the vector of appropriate size of all ones.
Letting $ \chi_i $ equal the entry with maximal magnitude in $ x $ (which equals $ \| x \|_\infty$), 
\[
\chi_i = \pm  \period \beta_0 \beta_1 \cdots \times 2^t,
\]
where $ \beta_0 = 1 $, \response{and} $ \sigma = 2^{-t } $ in our discussion. 
Hence
$
\| x  \|_\infty \leq \sigma \leq 2 \| x  \|_\infty
$.
We conclude that
$ \sigma \leq 2 \| x \|_\infty $
so that
\[
 \vert \delta\!\chi  \sigma \vert \leq 2 \machepsnew \| x \|_\infty \vec \jmath .
 \]
 Again, this is pessimistic given that the last split of the vector, $ x_3 $ carries a full FP64 remainder.  In line with the reasoning that resulted in (\ref{eqn:error_scalar}), this yields
 \[
 \widehat x =
( x_0 + x_1 \sigma_1 + x_2 \sigma_2 + x_3 \sigma_3  ) \sigma 
+ \delta\!x \sigma
,
 \]
where
\begin{equation}
\label{eqn:sigma}
\vert ~ \delta\!x  ~ \vert \sigma  \leq  \left( 2\machepsnew \| x \|_\infty \vec j \right)
.
\end{equation}
Here $ \macheps = 2^{-53} $ and $ \machepsnew = 2^{-117} $ when  $ c_0 = 22 $, $ c_1 = c_2 = 21 $, and $ c_3 = 53 $.

\subsection{Cascading matrices and \gemm}

\label{sec:CascadingGemm}

We are now ready to discuss how to cascade FP64x2 matrices $ \widehat A $ and $ \widehat B $ and approximately compute
$ \widehat A \widehat B $.
Here we restrict the ``inner dimension'' of the matrices to equal the $ k $ discussed in Section~\ref{sec:cascading_dot} since later we will see that high-performance implementations of \gemm\ block the operands in a way that imposes that restriction.  Such a \gemm\ is often referred to as a {\em rank-k update}.

Now
\[
\begin{array}{r@{\hspace{2pt}}c@{\hspace{2pt}}r@{}r@{\hspace{2pt}}c@{\hspace{2pt}}r@{}r@{\hspace{2pt}}c@{\hspace{2pt}}r@{}r@{\hspace{2pt}}c@{\hspace{2pt}}r@{}l@{\hspace{2pt}}l@{\hspace{2pt}}}
\widehat A &=& \Sigma ~ ( A_0 &  &+& \sigma_1 & A_1  &+&  \sigma_2 & A_2 &+& \sigma_3  & A_3 & )  \\
  \widehat B &=&( B_0 &  &+& \sigma_1 & B_1  &+& \sigma_2 & B_2  &+& \sigma_3 & B_3  & ) ~ T  ,
 \end{array}
\]
where $ \Sigma $ and $ T $ are the diagonal matrices%
\footnote{
Recall that multiplying a matrix from the left by a diagonal matrix scales its rows by the corresponding element of the diagonal matrix.
Similarly,
multiplying a matrix from the right by a diagonal matrix scales its columns.
}
that appropriately scale the rows of $ \widehat A $ (which become $ \widehat x $ in the dot products discussed in Section~\ref{sec:cascading_dot}) and columns of $ \widehat B $
(which become $ \widehat y $ in the dot products discussed in Section~\ref{sec:cascading_dot}).
The cascading of rows of $ \widehat A $ and columns of $ \widehat B $ follow exactly the discussion in the last section, since elements of the product are computed with dot products.

Then
{
\begin{equation}
    \label{eqn:Gemm16}
\begin{array}{l@{\hspace{2pt}}l@{\hspace{2pt}}r@{\hspace{2pt}}r@{\hspace{2pt}}r@{\hspace{2pt}}r@{\hspace{2pt}}r@{\hspace{2pt}}r@{\hspace{2pt}}r@{\hspace{2pt}}r@{\hspace{2pt}}r@{\hspace{2pt}}r@{\hspace{2pt}}r@{\hspace{2pt}}r@{\hspace{2pt}}r}
\widehat A  \widehat B \approx \Sigma ~ ( &
 A_0 B_0  &+&
\sigma_1  A_0 B_1  &+&
\sigma_2 A_0 B_2  &+&
\sigma_3 A_0 B_3   \\
&&+&
\sigma_1 A_1 B_0   &+&
\sigma_1^2 A_1 B_1  &+&
\sigma_1 \sigma_2 A_1 B_2  &+&
\sigma_1 \sigma_3 A_1 B_3 &+& \\
&&&&+&
\sigma_2 A_2 B_0  &+&
\sigma_2 \sigma_1 A_2 B_1 &+&
\sigma_2^2 A_2 B_2  &+&
\sigma_2 \sigma_3 A_2 B_3 &+& \\
&&&&&&+&
\sigma_3 A_3 B_0 &+&
\sigma_3 \sigma_1 A_3 B_1 &+&
\sigma_3 \sigma_2 A_3 B_2 &+&
\sigma_3^2 A_3 B_3 & )~ T . \\[-0.1in]
&
\begin{array}[t]{@{}c@{}}
\underbrace{
\phantom{A_0 A_0}
}
\\
\mbox{bin 0}
\end{array}
&&
\begin{array}[t]{@{}c@{}}
\underbrace{
\phantom{A_0 B_0 \tau_0 }
}
\\
\mbox{bin 1}
\end{array}
&&
\begin{array}[t]{@{}c@{}}
\underbrace{
\phantom{A_0 A_0 \sigma_0 \tau_0}
}
\\
\mbox{bin 2}
\end{array}
&&
\begin{array}[t]{@{}c@{}}
\underbrace{
\phantom{A_0 A_0 \sigma_0 \tau_0}
}
\\
\mbox{bin 3}
\end{array}
&&
\begin{array}[t]{@{}c@{}}
\underbrace{
\phantom{A_0 A_0 \sigma_0 \tau_0}
}
\\
\mbox{bin 4}
\end{array}
&&
\begin{array}[t]{@{}c@{}}
\underbrace{
\phantom{A_0 A_0 \sigma_0 \tau_0}
}
\\
\mbox{bin 5}
\end{array}
&&
\begin{array}[t]{@{}c@{}}
\underbrace{
\phantom{A_0 A_0 \sigma_0 }
}
\\
\mbox{bin 6}
\end{array}
\end{array}
\end{equation}
}%
The point is that one can approximately compute the FP64x2 \gemm\ in terms of 16 FP64 \gemm s, extending the observations about dot products.
 This means that
\begin{itemize}
    \item
    $ D_0 $, $ D_1 $, and $ D_2 $ are picked as discussed in Section~\ref{sec:cascading_dot}.
    \item 
    Elements within rows of $ \widehat A $ must be split conformally.
    \item
   Elements within columns of $ \widehat B $ must be split conformally.
\end{itemize} 
The terms in bins 0--2 are computed exactly.  Unless ``catastrophic cancellation'' happens, the contributions of those bins will contribute the most significant bits of the results. 

Again, bins 3--6 can be combined, since they contribute to lower order terms.
\[
\begin{array}{r@{\hspace{2pt}}l@{\hspace{2pt}}l@{\hspace{2pt}}r@{\hspace{2pt}}r@{\hspace{2pt}}r@{\hspace{2pt}}r@{\hspace{2pt}}r@{\hspace{2pt}}@{\hspace{2pt}}r@{\hspace{2pt}}r@{\hspace{2pt}}r@{\hspace{2pt}}r@{\hspace{2pt}}r@{\hspace{2pt}}r}
\mbox{bins 3--6} =  &
 \sigma_3 A_0 B_3    
&+&\\
& \sigma_1 \sigma_2  A_1 B_2 &+&
\sigma_1 \sigma_3 A_1 B_3  &+& \\
& \sigma_2 \sigma_1 A_2 B_1 &+&
\sigma_2^2 A_2 B_2 &+&
\sigma_2 \sigma_3 A_2 B_3 &+& \\
&
\sigma_3 A_3 B_0 &+&
\sigma_3 \sigma_1 A_3 B_1 &+&
\sigma_3 \sigma_2  A_3 B_2 &+&
\sigma_3^2 A_3 B_3 \\[-0.1in]
&
\begin{array}[t]{@{}c@{}}
\underbrace{
\phantom{A_0 B_0 \sigma_0 \tau_0}
}
\\
\mbox{bin 3}
\end{array}
&&
\begin{array}[t]{@{}c@{}}
\underbrace{
\phantom{\chi_0 \psi_0 \phantom{\sigma_0 \tau_0}}
}
\\
\mbox{bin 4}
\end{array}
&&
\begin{array}[t]{@{}c@{}}
\underbrace{
\phantom{\chi_0 \psi_0 \phantom{\sigma_0 \tau_0}}
}
\\
\mbox{bin 5}
\end{array}
&&
\begin{array}[t]{@{}c@{}}
\underbrace{
\phantom{\chi_0 \psi_0 \sigma_0 \tau_0}
}
\\
\mbox{bin 6}
\end{array}
\\
= &
\multicolumn{9}{l}{
A_0  B_3 \sigma_3  
+
A_1 \sigma_1
\begin{array}[t]{c}
\underbrace{
\left( B_2  \sigma_2 +
B_3 \sigma_3 \right)
} \\
B_4 \sigma_2 
\end{array} 
+
A_2 \sigma_2
\begin{array}[t]{c}
\underbrace{
\left( B_1  \sigma_1 +
B_2 \sigma_2  +
B_3 \sigma_3 
\right) 
} \\
B_5 \sigma_1
\end{array}
}
\\
&
\multicolumn{9}{l}{
~~~~~~~~~~ +
\sigma_3 A_3
\begin{array}[t]{c}
\underbrace{
\left( B_0   +
\sigma_1 B_1 +
\sigma_2 B_2  
+
\sigma_3 B_3 
\right)
} \\
B_6 
\end{array}
}
\\
=
&
\multicolumn{9}{l}{
\sigma_3 A_0 B_3  
+
\sigma_1 \sigma_2 A_1 B_4 
+
\sigma_2 \sigma_1 A_2 B_5 
+
\sigma_3 A_3 B_6, 
}
\end{array}
\]
where $B_4 = B_2 + B_3\left({\sigma_3}/{\sigma_2}\right)$, $B_5 = B_1 + B_2\left({\sigma_2}/{\sigma_1}\right) + B_3\left({\sigma_3}/{\sigma_1}\right)$ and 
$B_6 = B_0 + B_1{\sigma_1} + B_2{\sigma_2} + B_3{\sigma_3}.$  

This also shows how the approximate FP64x2 \gemm\ can be cascaded into {\bf ten} rather than sixteen FP64 multiplications. 
Now 
(\ref{eqn:Gemm16})  can be depicted as the computation of 
\begin{eqnarray}
\nonumber
\lefteqn{
\left( \begin{array}{c}
\phantom{\sigma_0} A_0 \\ \hline
\sigma_1 A_1 \\ \hline
\sigma_2 A_2 \\ \hline
\sigma_3 A_3 
\end{array} \right)
\begin{array}{c}
\left( \begin{array}[t]{c | c | c | c || c | c | c }
B_0 & \sigma_1 B_1 & \sigma_2 B_2 & \sigma_3 B_3  & \sigma_2 B_4  & \sigma_1 B_5  & \phantom{\tau_0} B_6 
\end{array} \right) \\
~~ \\
~~ \\
~~
\end{array}
= } \\
\label{eqn:Gemm10}
& \left( \begin{array}{c | c | c | c || c | c | c }
 \phantom{\sigma_0} A_0 B_0 & \sigma_1 A_0 B_1 & \sigma_2 A_0 B_2 & \sigma_3 A_0 B_3 
 & \star & \star  & \star \\ \hline
\sigma_1 A_1 B_0 & \sigma_1^2 A_1 B_1 & \star & \star &
\sigma_1 \sigma_2 A_1 B_4 & \star & \star \\ \hline
\sigma_2 A_2 B_0 & \star & \star  & \star &
\star & \sigma_2 \sigma_1 A_2 B_5 & \star \\ \hline
\star & \star & \star & \star &
\star & \star & \sigma_3 A_3 B_6
\end{array}
\right)
\end{eqnarray}
followed by the summation of the elements or the appropriate contributions into bins and then the accumulation of those, in appropriate precision.   Here the $\star$ entries do not need to be computed.

We finish by noting that from (\ref{eqn:sigma}) it follows that
\[
\widehat A = \Sigma
( A_0 + \sigma_1 A_1 + \sigma_2 A_2 + \sigma_3 A_3 ) + \Sigma \Delta\!\!A 
,
\]
where 
\[
 \vert \Sigma \Delta\!\!A \vert 
\leq 
 2 \machepsnew D_A J 
\mbox{~with~}
D_A = {\rm diag}( 
\max_j \vert \alpha_{0,j} \vert, 
\max_j \vert \alpha_{1,j} \vert ,
\ldots
)
.
\]
Here, in $ \Delta\!\!A $ the $ \Delta $ touches the $ A $ to indicate this is one symbol to represent the error matrix;
$ \vert ~ \cdot ~ \vert $ returns the element-wise absolute value; and 
$ \leq $ is an element-wise comparison  $ J $ is the matrix of appropriate size of all ones.

If $ \delta\!\alpha_{i,j} $ denotes the error incurred when cascading element $ \alpha_{i,j} $, 
this captures that 
\[ \vert \delta\!\alpha_{i,j} \vert \leq 
 2 \machepsnew \max_{j} \vert \alpha_{i,j} \vert .
\]
An interpretation of this is that storing a real-valued matrix as a FP64x2 matrix incurs an element-wise relative error.  Cascading a FP64x2 matrix $ \widehat A $ with our method incurs an error in a given element that is proportional to the magnitude of the largest element in the same row for $ \widehat A $.  
However, because the last split of a given element captures the remainder of what it left when the other splits are subtracted from the original FP64x2 number, when those  other splits equal zero, what is left is a FP64 approximation and hence the error is never worse than what is incurred when storing from FP64x2 to FP64.
\response{This feature is missing from competitive algorithms where the worst error found from converting FP64x2 into fixed-point might be arbitrarily large if one doesn't have enough splits.}

\subsection{Blocking for the $ k $ dimension}

\response{
We are cascading the matrices into splits for both  error-free
(in effect, fixed point implemented in floating point) and error incurring (floating point)  multiplication. Since we force the computation of only the first few bins of matrix multiplication to be error-free, we must pay attention to the inner $k$ dimension of matrices. 
In theory, as $k$ increases, the number of splits would also increase to ensure the first few bins are computed error-free. However, importantly, high-performance \gemm~ implementation already requires the matrices be blocked for the registers as well as the various cache levels~\cite{Goto1}. We know that the $k$ dimension of matrix $A$ and $B$ is blocked by $k_C$, a parameter that is determined based on the size of the L1 cache~\cite{BLIS4}. 

We take advantage of the $k$-blocking done for high performance, and use the same value of $k_C$ to determine how many bits we can track in each split to ensure the computations of the first few bins are error-free. This in turn allows us to fix the number of splits of $A$ and $B$ to four and the total number of \gemm s to ten. In our proposed method, the number of splits is independent of the input matrix sizes, while in~\cite{OzakiScheme} it is matrix-size dependent. 

{One could manually block OzBLAS results in the $k$ dimension as suggested in~\cite{OzakiScheme}, but this ignores one of the key difference between our strategies. OzBLAS requires more splits (and the number of multiplies grows quadratically with the splits, so the number of multiplies that need to be performed is higher) since their method attempts to reduce the errors in all the bins. While this computationally intense strategy produces more accurate results, we will show that it can suffice to allow approximations in the final bins for much faster {and predictable} time.}
}

\subsection{Forward Error Analysis}

Consider a simple worst-case example of how cascading can 
reduce accuracy in a dot product, the building block of a 
matrix multiplication.
Let 
\[
\widehat x = \left( \begin{array}{c}
1 \\
\epsilon
\end{array}
\right)
\mbox{ and }
\widehat y = \left( \begin{array}{c}
0 \\
1
\end{array}
\right).
\]
If $ \epsilon < 2^{-D_2} $ (in our case $D_2=64$), then 
cascading $\widehat x $ yields
\[
\begin{array}[t]{c}
\underbrace{
\left( \begin{array}{c}
1 \\
\epsilon
\end{array}
\right)
}
\\
\widehat x
\end{array}
\approx
\begin{array}[t]{c}
\underbrace{( 1 )}
\\
\sigma
\end{array}
(
\begin{array}[t]{c}
\underbrace{
\left( \begin{array}{c}
1 \\
0
\end{array}
\right)
}
\\
x_0
\end{array}
+
\sigma_1
\begin{array}[t]{c}
\underbrace{
\left( \begin{array}{c}
0 \\
0
\end{array}
\right)
}
\\
x_1
\end{array}
+
\sigma_2
\begin{array}[t]{c}
\underbrace{
\left( \begin{array}{c}
0 \\
0
\end{array}
\right)
}
\\
x_2
\end{array}
+
\sigma_3
\begin{array}[t]{c}
\underbrace{
\left( \begin{array}{c}
0 \\
{\rm fl}_{\rm FP64}(\epsilon / \sigma_3 )
\end{array}
\right)
}
\\
x_3
\end{array}
),
\]
where $ {\rm fl}_{\rm FP64}( \zeta ) $ equals the FP64 approximation to $ \zeta $. Notice that $x_3$ can store at most 53-bits of $\epsilon$. Similarly, cascading $\widehat{y}$ yields \[
\begin{array}[t]{c}
\underbrace{
\left( \begin{array}{c}
0 \\
1
\end{array}
\right)
}
\\
\widehat y
\end{array}
\approx
\begin{array}[t]{c}
\underbrace{( 1 )}
\\
\tau
\end{array}
(
\begin{array}[t]{c}
\underbrace{
\left( \begin{array}{c}
0 \\
1
\end{array}
\right)
}
\\
y_0
\end{array}
+
\sigma_1
\begin{array}[t]{c}
\underbrace{
\left( \begin{array}{c}
0 \\
0
\end{array}
\right)
}
\\
y_1
\end{array}
+
\sigma_2
\begin{array}[t]{c}
\underbrace{
\left( \begin{array}{c}
0 \\
0
\end{array}
\right)
}
\\
y_2
\end{array}
+
\sigma_3
\begin{array}[t]{c}
\underbrace{
\left( \begin{array}{c}
0 \\
0
\end{array}
\right)
}
\\
y_3
\end{array}
).
\]
Then
$
\widehat x^T \widehat y = {\rm fl}_{\rm FP64x2}(\epsilon) = \epsilon
$ when computed in FP64x2 arithmetic (incurring no error in this very special case)
and $ \widetilde x^T \widetilde y = 
{\rm fl}_{\rm FP64}( \epsilon )$  when computed as a cascaded dot product, incurring no error other than that incurred when cascading $ \widehat x $.  
In other words, the accuracy is no better than if the computation had been performed in FP64 arithmetic because the entire vector $ \widehat x $ is conformally split. In the worst case, if $\epsilon < 2^{-117}$, then we potentially obtain zero bits of accuracy. 
This illustrates the need for a careful forward error analysis. The casual reader may want to skip to Section~\ref{sssec:errordiscussion} to obtain an overview. 
\subsubsection{Preparation}

We review some classic results from numerical linear algebra~\cite{Higham:2002:ASN,Bientinesi:2011:GMS:2078718.2078728}.  Here $ \mathbb{F} $ and $ \mathbb{F}^k $ denote the sets of floating point numbers and vectors, respectively.  The machine epsilon for the precision is given by $ \macheps $ so that if $ \chi \in \mathbb{R} $, then $ {\rm fl}( \chi ) = \chi( 1 + \epsilon ) $, where $ \vert \epsilon \vert \leq \macheps $.

\begin{definition}[Standard Computational Model]
Let $ \chi, \psi \in \mathbb{F} $.  Then 
\begin{eqnarray*}
     \chi \oplus \psi &=& {\rm fl}( \chi + \psi ) = ( \chi + \psi )( 1 + \epsilon_+ ),
     \mbox{ where } \vert \epsilon_+ \vert \leq \macheps, \mbox{ and }
\\
\chi \otimes \psi &=& {\rm fl}( \chi \times \psi ) = ( \chi \times \psi )( 1 + \epsilon_\times ),
     \mbox{ where } \vert \epsilon_\times \vert \leq \macheps.
\end{eqnarray*}
\end{definition}

\begin{theorem}
    Let $ \epsilon_0, \cdots , \epsilon_{n-1} $ satisfy  $ \vert \epsilon_i \vert \leq \macheps $. Then 
there exists a 
$ \theta_n $ such that
$
( 1 + \theta_n ) = 
( 1 + \epsilon_0 ) \cdots ( 1 + \epsilon_{n-1} )  $, where $
\vert \theta_n \vert \leq \gamma_n \macheps $ with $ \gamma_n = n \macheps / ( 1 - n \macheps ) \approx n \macheps $ (if $ n $ is relatively small).
\end{theorem}

{\bf Important:} As usual in these analyses, we are going to use $ (1 + \theta_n) $ to captures  how many times (at most) a value has been affected by round-off, making it a ``number of round-off occurances'' measure.  In particular, multiple instance of $ \theta_n $ may not equal the same quantity, since they may represent the accumulation of different epsilons. 
This also means that
$
( 1+ \theta_m )(1+\theta_n) = ( 1 + \theta_{m+n} ) 
$
and
if $ m \leq n $, then an occurrence of $ \theta_m $ can be replaced by $ \theta_ n $.

\begin{theorem}
Let $ v, w \in \mathbb{R}^k $.  Then
$
{\rm fl}( v^T w ) = v^T ( I + \Theta_k ) w
$, where $ \Theta_k = 
{\rm diag}( \theta_k, \theta_k, \theta_{k-1}, \ldots ) 
$
and
$
\vert  {\rm fl}( v^T w ) - v^T w \vert \leq \gamma_k  \vert v \vert^T \vert w \vert \approx k \macheps \vert v \vert^T \vert w \vert .$
\end{theorem}

In our analysis, $ \Theta_n = {\rm diag}(\theta_n, \theta_n, \theta_{n-1} , \ldots ) $ is a matrix ``of appropriate size'' in the context in which it occurs.

\begin{corollary}
Let $ v_i, w_i \in \mathbb{R}^k $.  Then
$
{\rm fl}( \sum_{i=0}^{p-1} v_i^T w_i ) = \sum_{i=0}^{p-1} v_i^T ( I + \Theta_{k+p-1} ) w_i
$. 
\end{corollary}

\begin{proof}
W.l.o.g., we assume the summation is computed in the order $ i  = 0, 1, \ldots $.
   \begin{eqnarray*}
       \lefteqn{{\rm fl}( \sum_{i=0}^{p-1} v_i^T w_i )  
        = 
       (
       \cdots 
       (
       {\rm fl}( v_0^T w_0 )
       \oplus
       {\rm fl}( v_1^T w_1 )
       )
       \oplus
       \cdots
       )
       \oplus
       {\rm fl}( v_{p-1}^T w_{p-1} ) }
       \\
             & = & 
       (
       \cdots 
       (
       v_0^T ( I + \Theta_k ) w_0 
       +
    v_1^T ( I + \Theta_k ) w_1 
       )
       (1 + \epsilon_1 )
       +
       \cdots
       )
       +
 v_{p-1}^T ( I + \Theta_k ) w_{p-1}
 (1 + \epsilon_{p-1} )
       \\
             & = & 
v_0^T ( I + \Theta_k ) w_0 \left[ ( 1 + \epsilon_{1} )\cdots ( 1 + \epsilon_{p-1}  ) \right]
+
v_1^T ( I + \Theta_k ) w_1 \left[ ( 1 + \epsilon_{1} )\cdots ( 1 + \epsilon_{p-1}  ) \right]
+
\cdots 
\\
& & ~~~~~~
+~
v_{p-1}^T ( I + \Theta_k ) w_{p-1} ( 1 +  \epsilon_{p-1}  ) 
\\
&  = &
v_0^T ( I + \Theta_k ) w_0 ( 1 + \theta_{p-1} )
+
v_1^T ( I + \Theta_k ) w_1 ( 1 + \theta_{p-1} )
+
\cdots 
+
v_{p-1}^T ( I + \Theta_k ) w_{p-1} ( 1 + \theta_1 )
\\
&  = &
v_0^T ( I + \Theta_k ) w_0 ( 1 + \theta_{p-1} )
+
v_1^T ( I + \Theta_k ) w_1 ( 1 + \theta_{p-1} )
+
\cdots 
+
v_{p-1}^T ( I + \Theta_k ) w_{p-1} ( 1 + \theta_{p-1} )
\\
&  = &
v_0^T ( I + \Theta_{k+p-1} ) w_0 
+
v_1^T ( I + \Theta_{k+p-1} ) w_1 
+
\cdots 
+
v_{p-1}^T ( I + \Theta_{k+p-1} ) w_{p-1} 
\\
       &=& \sum_{i=0} v_i^T ( I + \Theta_{k+p-1} ) w_i.
   \end{eqnarray*} 
\end{proof}

Use of $ \macheps $, $ \theta_i $, and $ \gamma_i $ means FP64 arithmetic is employed while use of 
$ \machepsdd $, $ \widehat  \theta_i $, and $ \widehat  \gamma_i $
means FP64x2 arithmetic is used.

\subsubsection{Forward error}

Now let us turn to the error analysis of the dot product of row $ i $ of $ A $, $ x $,  with column $ j$ of $  B $, $ y $, where $ x, y \in \mathbb{R}^{k} $.  Since \mbox{bin 0}, \mbox{bin 1}, and \mbox{bin 2} are computed exactly, we find that
\begin{eqnarray*}
{\rm fl}_{\rm casc}(  x^T  y )
&=&
{\rm fl}_{\rm casc}( \sigma \tau (\mbox{bin 0} +
\mbox{bin 1} +
\mbox{bin 2} +
\mbox{bin 3-6})
 )
\\
&=&
\sigma \tau ( \mbox{bin 0} \oplus (
\mbox{bin 1} \oplus
( \mbox{bin 2} \oplus
{\rm fl}_{\rm FP64}( \mbox{bin 3-6} ) )) ).
\end{eqnarray*}
where the $ \oplus $ are performed in FP64x2.
Since bins 0 through 2 are computed exactly,
\begin{eqnarray}
\nonumber
\lefteqn{
\vert
{\rm fl}_{\rm casc}( x^T y )
-
 x^T y
\vert
=
\vert
\sigma \tau ( (\mbox{bin 0}) \oplus (
(\mbox{bin 1}) \oplus
( (\mbox{bin 2}) \oplus
{\rm fl}_{\rm FP64}( \mbox{bin 3-6} ) )) )
-
x^T y
\vert}
\\
\nonumber
& = &
\vert
\sigma \tau (
(\mbox{bin 0}) + (
(\mbox{bin 1}) +
( (\mbox{bin 2}) +
{\rm fl}_{\rm FP64}(\mbox{bin 3-6}
) ) (1 + \widehat \epsilon_2 ) )
(1 + \widehat \epsilon_1 ))
(1 + \widehat \epsilon_0 ) )
-
x^T y
\vert
\\
\nonumber
& = &
\vert
\sigma \tau ( 
(\mbox{bin 0})
( 1 + \widehat \epsilon_0 ) + 
(\mbox{bin 1}) 
( 1 + \widehat \theta_2 )
+
(\mbox{bin 2}) ( 1 + \widehat \theta_3 )
+
(\mbox{bin 3-6}) 
( 1 + \widehat \theta_3 ) )
-
x^T y
\vert
\\
\nonumber
& = &
\sigma \tau
\vert (
\mbox{bin 0})
 \widehat \epsilon_0  + 
(\mbox{bin 1} )
 \widehat \theta_2 
+
(\mbox{bin 2} ) \widehat \theta_3 
+
{\rm fl}_{\rm FP64}(
\mbox{bin 3-6} ) ( 1 + \widehat \theta_3 )
-
\mbox{bin 3-6}
\vert
\\
\label{eqn:casc1}
& \leq &
\sigma \tau
( 
 \vert
\mbox{bin 0}
 \vert
  \vert
   \widehat \epsilon_0 
  \vert
  + 
 \vert
 \mbox{bin 1}  \vert
  \vert
 \widehat \theta_2 
  \vert
+
 \vert
 \mbox{bin 2} 
  \vert
   \vert
   \widehat \theta_3
    \vert
+
 \vert
{\rm fl}_{\rm FP64}(
\mbox{bin 3-6} ) 
( 1 + \widehat \theta_3 )
-
\mbox{bin 3-6}
\vert ).
\end{eqnarray}

In our analysis, it will become important that splits $ x_i $ and $ y_i $ all have entries that are bounded in magnitude by one.  This is captured by $ \vert x_i \vert \leq \vec \jmath $ and 
$ \vert y_i \vert \leq \vec \jmath $.
Also, because of the special way the vectors are split, ranges stored in the splits do not overlap, 
\[
\vert x \vert =
\vert \sigma( x_0 + \sigma_1 x_1 + \sigma_2 x_2 +
\sigma_3 x_3 ) \vert
=
\vert  \sigma x_0  \vert  + \vert  \sigma \sigma_1 x_1 \vert  + \vert  \sigma \sigma_2 x_2 \vert +
\vert \sigma \sigma_3 x_3 \vert
\]
so that $ \sigma \vert x_0 \vert \leq \vert x \vert $ and $ \tau \vert y_0 \vert \leq \vert y \vert $.

Now,
\begin{eqnarray*}
     \vert \mbox{bin 0} \vert \vert
\widehat \epsilon_0
\vert &=&
\vert x_0^T y_0 \vert
\vert
\widehat \epsilon_0
\vert
\leq
\vert x \vert^T \vert y \vert
\vert
\widehat \epsilon_0
\vert
\leq
\vert x \vert^T \vert y \vert
\machepsdd
, \\
     \vert \mbox{bin 1} \vert \vert
\widehat
\theta_2
\vert &=&
\vert \sigma_1 x_0^T y_1
+
\sigma_1 x_0^T y_1
\vert
\vert
\widehat
\theta_2
\vert
\leq
\sigma_1 ( 
\vert x_0 \vert^T \vert y_1 \vert
+
\vert x_1 \vert^T \vert y_0 
\vert
)
\vert
\widehat
\theta_2
\vert
\leq
2 \sigma_1
\vec \jmath ^T
\vec \jmath
\gamma_2
\approx
4 \sigma_1 k 
 \machepsdd , \mbox{ and}
\\
     \vert \mbox{bin 2} \vert \vert
\widehat
\theta_3
\vert &=&
\vert \sigma_2 x_0^T y_2
+
\sigma_1^2 x_1^T y_1
+
\sigma_2 x_2^T y_0
\vert
\vert
\widehat
\theta_3
\vert
\leq
( 2\sigma_2 + \sigma_1^2 )
k
\widehat
\gamma_3
\approx
3 ( 2\sigma_2 + \sigma_1^2 )
k
\machepsdd.
\end{eqnarray*}
This leaves us to examine error in \mbox{bin 3-6}.  For simplicity, we use the original ten terms in~(\ref{eqn:xy}) for computing \mbox{bin 3--6}, which can be written as
\[
\mbox{bin 3-6} =
\sum_{(i,j) \in S}
\sigma_i \sigma_j x_i^T y_j,
\]
where $ \sigma_0 = 1 $ and $ S $ (with $ \vert S \vert = 10 $ elements) is the set of indices that occurs in the computation of \mbox{bin 3} through \mbox{bin 6} in~(\ref{eqn:xy}).  
Now
\begin{eqnarray*}
\lefteqn{
\vert {\rm fl}_{\rm FL64}( \mbox{bin 3-6} ) ( 1 + \widehat \theta_3 ) -
\mbox{bin 3-6}  \vert
 = 
\vert
{\rm fl}_{\rm FL64}( \sum_{(i,j) \in S}
\sigma_i \sigma_j x_i^T y_j )
( 1 + \widehat \theta_3 ) -
\sum_{(i,j) \in S}
\sigma_i \sigma_j x_i^T y_j
\vert} \\
& = &
\vert
\sum_{(i,j) \in S}
\sigma_i \sigma_j x_i^T (I + \Theta_{k+\vert S \vert -1} ) y_j 
( 1 + \widehat \theta_3 ) -
\sum_{(i,j) \in S}
\sigma_i \sigma_j x_i^T y_j
\vert 
=
\vert
\sum_{(i,j) \in S}
\sigma_i \sigma_j x_i^T \Theta_{k+9} y_j
\widehat \theta_{3} \vert
\\
& = &
\vert
\sum_{(i,j) \in S}
\sigma_i \sigma_j x_i^T \Theta_{k+12} y_j
 \vert
 \leq 
\max_{(i,j) \in S}
(\sigma_i \sigma_j)
\sum_{(i,j) \in S}
 \vert x_i^T \Theta_{k+12} y_j
 \vert
 \\
& \leq &
\max_{(i,j) \in S}
(\sigma_i \sigma_j)
\sum_{(i,j) \in S}
 ( \vert x_i \vert ^T \vert y_j
 \vert )
 \gamma_{k+12} 
 \leq
 \max_{(i,j) \in S}
(\sigma_i \sigma_j)
\sum_{(i,j) \in S}
 ( \vec \jmath ^T \vec \jmath )
 \gamma_{k+12} 
 \leq
\vert S \vert \times 2^{-D_2} k \gamma_{k+12}
 \\
 &\approx&
 10 \times 2^{-D_2} k (k+12) \macheps
 \approx
 10 \times 2^{-D_2} k^2 2^{-D}
 \approx
 10 k^2 \machepsnew.
\end{eqnarray*}
Here $ k + 12 \approx k $ since $ k $ is typically in the $ 256 $ range%
\footnote{Keep in mind that some equalities come from the fact that $ \theta_i$s have special meaning.}.

Putting it all together, we find that
\begin{eqnarray}
\vert
{\rm fl}_{\rm casc}(  x^T  y )
-
 x^T  y
\vert 
&\leq& \mbox{(\ref{eqn:casc1})}
\approx
\vert x \vert^T \vert y \vert \machepsdd +
\sigma \tau 
( 
4 \sigma_1 k 
 \machepsdd +
3 ( 2\sigma_2 + \sigma_1^2 )
k
\machepsdd 
+
10 k^2 \machepsnew ) \nonumber
\\ 
& \approx &
\vert x \vert^T \vert y \vert \machepsdd
+
\sigma \tau 
10 k^2 \machepsnew \nonumber \\
&\leq &
\vert x \vert^T \vert y \vert \machepsdd
+
4 \times 10 k^2 \machepsnew 
\| x \|_\infty \| y \|_\infty,
\label{eq:cascerr}
\end{eqnarray}
since 
$ \sigma_1 $ and $ \sigma_1^2 $ are very small, and 
$ \sigma \leq 2 \| x \|_\infty$ and $\tau \leq 2 \| y \|_\infty $ (so $\sigma \tau \leq 4 \| x \|_\infty \| y\|_\infty$.)

As noted, this analysis uses the ten dot products from~(\ref{eq:dot}) to compute bins 3-6.  If the 
four dot product expression from~(\ref{eqn:xy}) were used, a similar bound would emerge.
The analysis
ignores the fact that a matrix multiplication is staged as a sequence of rank-k updates where each rank-k limits the choice of $ k $ to, for this work, 256.  This translates to each entry in $ C $ being computed as several dot products with vectors of size $ k $.  Such accumulation is performed in FP64x2 arithmetic, which would limit the additional loss of accuracy that is incurred.

{
In the setting of a cascaded matrix multiplication,  the analysis of the dot product gives us insight into  what forward error to expect in the dot product of rows of $ A $ with columns of $ B $.  Letting $\widetilde a_i^T $ and $ b_j $ equal the $ i$th row  of $ A $ and $ j $th column of $B $, (\ref{eq:cascerr}) is summarized as  (approximately)
\begin{eqnarray}
\vert
{\rm fl}_{\rm casc}(  A  B )
-
A B
\vert 
\leq
\machepsdd \vert A  \vert  \vert B \vert 
+
4 \times 10 k^2 \machepsnew 
\left( \begin{array}{c}
\| \widetilde a_0 \|_\infty \\
\| \widetilde a_1 \|_\infty \\
\vdots  \\
\| \widetilde a_{m-1} \|_\infty 
\end{array} \right)
\left( \begin{array}{c}
\| b_0 \|_\infty \\
\| b_1 \|_\infty \\
\vdots  \\
\| b_{n-1} \|_\infty 
\end{array} \right)^T.
\end{eqnarray}
Taking the Frobenius norm of both sides  we get   (approximately)
\begin{eqnarray}
\|
{\rm fl}_{\rm casc}(  A  B )
- A B
\|_F 
\leq
 \machepsdd \| A \|_F \| B \|_F
+
4 \times 10 k^2 \machepsnew 
\| A \|_F \| B \|_F .
\end{eqnarray}
}

\subsubsection{Discussion}
\label{sssec:errordiscussion}
The forward error in a cascaded dot product given by~\eqref{eq:cascerr} depends on how large $\vert  x \vert^T \vert  y \vert$ is compared to $\| x \|_\infty \| y \|_\infty$. For example, consider vectors $x, y\in \mathbb{R}^k$, where 
\[x = \left( \begin{array}{c}
     1  \\
     \epsilon \\
     \epsilon \\
     \vdots \\
     \epsilon
\end{array}\right) \mbox{ and } y = \left(\begin{array}{c} \epsilon \\ 1 \\ \epsilon \\\vdots \\ \epsilon \end{array}\right).\]
Now, if $1 \gg \epsilon > 0$, then $\| x \|_\infty \| y \|_\infty = 1$, but $\vert  x \vert^T \vert  y \vert = 2\epsilon + (k-2)\epsilon^2.$ Thus, the final bound for the forward error in a cascaded dot product is possibly worse than the bound for dot product in FP64x2 arithmetic, which is given by
\begin{equation}
    \label{eqn:fp64x2}
 \vert {\rm fl}_{\rm FP64x2}(  x^T  y ) -  x^T  y \vert  \leq 
 \gamma_k \vert  x \vert^T \vert  y \vert
 \approx 
 k \machepsdd \vert  x \vert^T \vert  y \vert.
 \end{equation}

\noindent Since bin 0--2 are computed error free, if result in bin 0 is not zero then 
 \[\vert  x \vert^T \vert  y \vert \geq 2^{-21} \| x \|_\infty 2^{-21} \|  y \|_\infty = 2^{-42} \| x \|_\infty \|  y \|_\infty .\] Substituting this constraint and $k=2^8 (=256)$, $\machepsdd = 2^{-106}$,  $\machepsnew = 2^{-117}$ in~\eqref{eq:cascerr}, we get
 \begin{eqnarray}
 \nonumber
    \vert
{\rm fl}_{\rm casc}(  x^T  y )
-
 x^T  y
\vert & \leq &  \vert x \vert^T \vert y  \vert \times 2^{-106}
+
4 \times 10 \times 2^8 \times 2^{-117} \times  2^{42} \times k  \times
\vert x \vert^T \vert y \vert  \\ 
     & \leq & \vert x \vert^T \vert y  \vert \times 2^{-106}
+
2^{-61} \times k  \times
\vert x \vert^T \vert y \vert.
\label{eq:fp642err}
\end{eqnarray} 
Similarly~\eqref{eqn:fp64x2} can be rewritten as \[\vert {\rm fl}_{\rm FP64x2}(  x^T  y ) -  x^T  y \vert  \leq 
 2^{-106} \times k \times \vert  x \vert^T \vert  y \vert.\]

 This analysis shows that the forward error bounds   
 \begin{itemize}
    \item are better than FP64 (53 bits mantissa);
     \item can be much better than double-double; and 
     \item  when bin 0 does not accumulate to zero, (\ref{eq:fp642err}) tells us the worst case error for cascading still leaves 61 correct bits in the mantissa.
 \end{itemize} 
 To gain further insight, we present numerical experiments in Section~\ref{sec:accuracy}. 

 An important aspect of cascading GEMM is that we can estimate all the condition numbers and potential worst-case errors after only computing bin 0, which is described in the next section.

\subsection{Detecting cancellation error }
\label{sec:detecting}

Dot products, and hence \gemm, can suffer severe cancellation errors which lead to less accurate results, especially for vectors that are nearly orthogonal. 
In most floating-point calculations this phenomenon is difficult to detect while it is happening unless one does very intrusive (not performance or power-friendly) checking before and after each FMA. To detect cancellation in \gemm, one usually looks at the final result of the dot product to see if it is small compared to the norm of the vectors that generated it. This is not typically done in practice, since it only is possible after the calculation is finished. Such detection of cancellation or computation of a condition number of the dot product could trigger the use of a different algorithm, but only after the previous algorithm completes. 

Cascading \gemm\ has an advantage, since the dot products are by definition broken into pieces. Detecting cancellation error can typically be done after just computing bin 0, well before the algorithm is completed (indeed, bin 0 represents only one tenth of the work we propose.) By detecting cancellation error earlier, other approaches can be be taken before one has completed much of the work. 

Furthermore, as previously noted, bins 0, 1 and 2 are computed error-free. If the result of a dot product contains a non zero in bin 0, then that result has  minimally 42 error-free bits from bins~1 and~2, along with whatever bits are nonzero in bin 0\footnote{Notice that there may be leading $D_2 \approx 64$ leading bits that are computed error-free, but because of leading zeros we can't say these are the most significant 64 bits.}, and the rest of the computation is done with regular FP64 DGEMM accuracy (up to another 53 bits), although not error-free. This is the primary reason why nonzeros in bin 0 yield results comparable with double-double. We also benefit from higher accuracy because the condition number of bins 3-6 is not necessarily as high as the condition number of the original input vectors%
\footnote{If one splits a real number then the splits other than the most significant one will tend to be more like a random number, which means that if dot products are formed with the vectors of such numbers, that dot product can be expected to have a reasonable condition number.}. In our ill-conditioned testing (Section~\ref{sec:accuracy}), when the conditioning grows, our method actually becomes even more favorable (again with the notable exception of zeros appearing in early bins.)

This leads to a strategy to check for zeros in bin 0 possibly before even finishing all ten \gemm s, and flag such results for being subject to cancellation errors. We further comment on the implementation of this in Section~\ref{sec:morecancellation}.

\remove{
For the case where $ k $ is small, if bin 0 has ANY non-zero bits, then unless bins 1 or 2 cancel that out (a possibility, but extremely unlikely to ever occur in practice), then that means we have at least $ D_2 - D_0 = 42 $ leading bits that are computed error-free, in addition to however many significant bits we found in bin 0, and in addition to the error-prone 53 additional bits we have in bin~3--6. 
{\bf I don't understand this:} This means that if bin 0 has exclusively non-zero values (every component has at least one non-zero bit), we are getting full 21-bit benefits from bin 1 and 2 (we know we are achieving 42 bits error-free) and the standard benefits from the following bins, and we're still likely to do better than FP64x2.
If $ k $ is large, we will discuss later how the multiplication can be broken down into multiple rank-k updates.  Whether   cancellation can be tracked in this case 
is an open question.
}

\section{Practical considerations}

\label{sec:practical}

We illustrate the viability of the scheme proposed in the last section by leveraging principles that underlie the BLAS-like Library Instantiation Software (BLIS) framework ~\cite{BLIS2,BLIS1}.
BLIS strives to enable optimizing performance while minimizing the amount of code that must be customized for a new architecture. Experience with many generations of CPUs has demonstrated that performance that rivals that of the best implementations can be achieved while keeping the effort manageable.

The proposed scheme for computing \gemm\ via cascaded matrices entails three phases:
\begin{description}
    \item [Phase one:] Splitting the matrices into its cascading components.  The technical details of our approach to this can be found in Appendix~\ref{app:A}.
    \item [Phase two:] Computing the various products within a bin exploiting a high-performing FP64 matrix multiplication. 
    \item[Phase three:] Adding across the bins in FP64x2 precision to obtain the final answer. 
\end{description}
We now examine how to integrate these into a high-performance implementation.

\subsection{High-performance implementation of \gemm}

\begin{figure}[tb!]
\begin{center}
\includegraphics[width=0.5\textwidth]{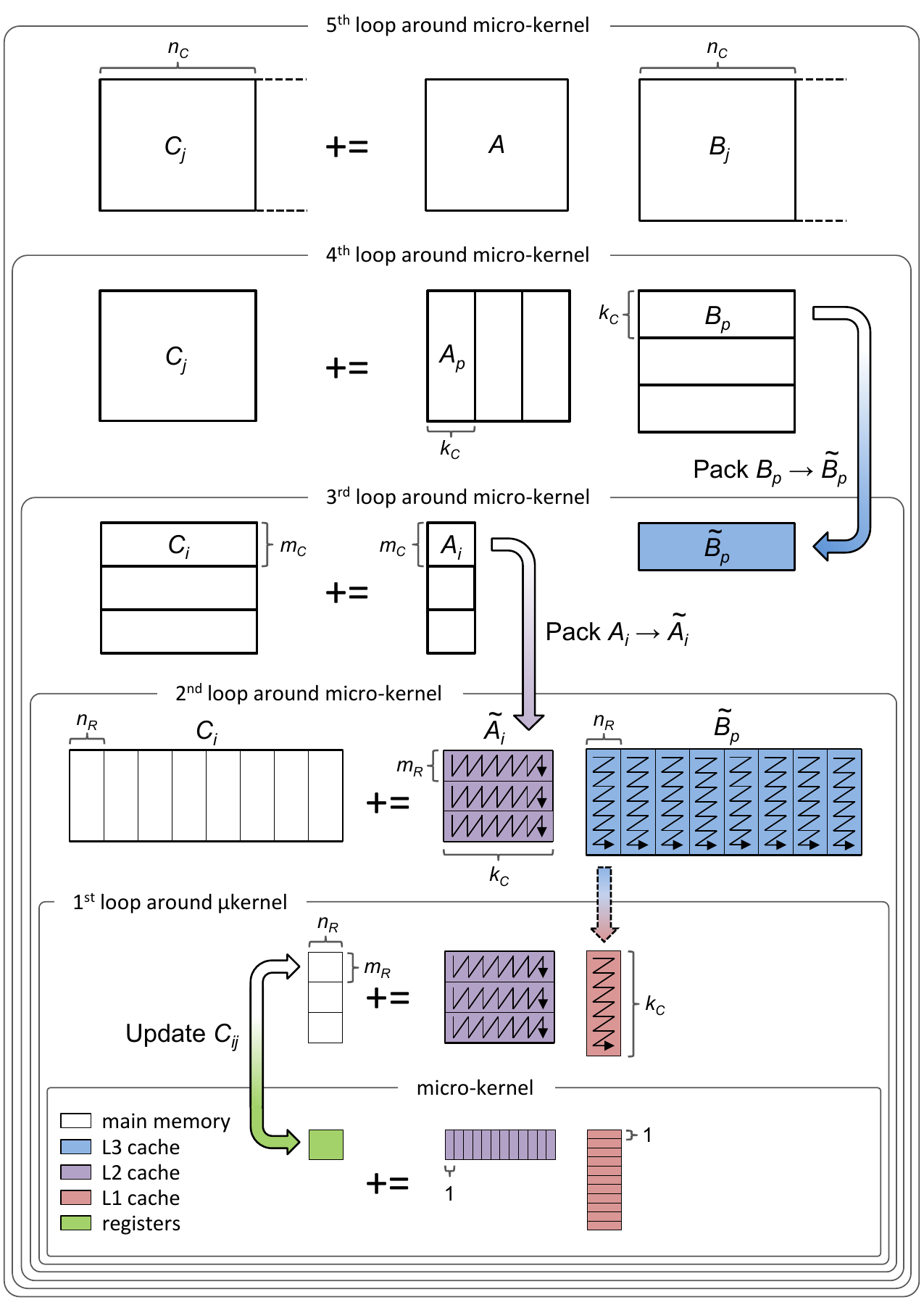}

\end{center}
\caption{
The BLIS refactoring of the GotoBLAS algorithm for \gemm\ as five loops around the microkernel.
This diagram, which is often used when explaining the fundamental techniques that underly the BLIS implementation of \gemm, was modified from a similar image first published in~\cite{BLIS5} and is used with permission.
}
\label{fig:BLIS}
\end{figure}

Most current high-performance implementations of GEMM for CPUs, including BLIS, implement Goto's Algorithm~\cite{Goto1}.
The BLIS instantiation of this algorithmn
exposes five loops (implemented in the C programming language) around a {\em micro-kernel} (implemented in assembly code or with intrinsic functions) as illustrated in Figure~\ref{fig:BLIS}.
The micro-kernel updates a $ m_R \times n_R $ {\em micro-tile} of $ C $, 
$
C_{ij}  
$,
by multiplying a $ m_R \times k_C $ {\em micro-panel} of $ A $ times a $ k_C \times n_R $ micro-panel of $ B $.  
On a typical achitecture, the micro-tile of $ C $ is kept in registers while the micro-panel of $ B $ is streamed from the L1 cache and the micro-panel of $ A $ is streamed from the L2 cache.  In order to access memory with stride one (consecutive access), blocks of $ A $ and row panels of $ B $ are packed (rearranged) at strategic points in the algorithm~\cite{GMH:92}.  Generally, blocking parameters can be determined analytically~\cite{BLIS4}.
One advantage of exposing the five loops around the micro-kernel in C is that it presents multiple loops where thread-level parallelism can be introduced~\cite{BLIS3,BLIS2}.

The original GotoBLAS implementation required the computation related to the two loops around the micro-kernel and the micro-kernel to be customized  in assembly code for a given architecture.   BLIS reduces this to just the micro-kernel.

\NoShow{
\subsection{Phase two: FP64x2 GEMM via ten FP64 GEMMs}
To calculate the various terms within the bin, one would require workspace to store these intermediate results. A simple approach would be to compute each of the ten products independently in their own workspace. However, this is an unnecessary waste of workspace. Since each product within a bin contains terms within the same range, these terms can be added together using FP64 arithmetic.  This reduces the workspace down to four from ten. The workspace can be reduced down to one by merging phase two and phase three in a way that each result is accumulated into one bin. This requires accumulation from the lowest order terms to the higher order terms. In addition, some of these accumulations must be done with FP64x2 arithmetic. 
}

\subsection{Simple: FP64x2 \gemm\ via ten FP64 \gemm s}
\label{sec:simple}

An important consequence of Goto's algorithm is that it achieves its high performance already for rank-k updates with $ k = k_C $ ($ = 256 $ on our target architecture discussed in Section~\ref{sec:perfexp}). 
This allows us to stage  
\gemm\ as a sequence of rank-k updates as shown in Figure~\ref{fig:gemmrankk}.
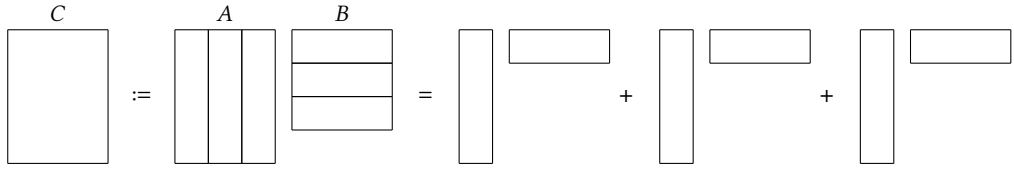
\begin{figure}

\begin{center}
{
\setlength{\unitlength}{0.175in}
\begin{picture}(30,5)
\put(0,0){\line(0,1){4}}
\put(3,0){\line(0,1){4}}
\put(0,0){\line(1,0){3}}
\put(0,4){\line(1,0){3}}
\put(1.5,4.5){\makebox(0,0){$C$}}

\put(4,2){\makebox(0,0){$:=$}}

\multiput(5,0)(1,0){3}{
\put(0,0){\line(1,0){1}}
\put(0,4){\line(1,0){1}}
\put(0,0){\line(0,1){4}}
\put(1,0){\line(0,1){4}
}
}
\put(6.5,4.5){\makebox(0,0){$A$}}

\multiput(8.5,1)(0,1){3}{
\put(0,0){\line(1,0){3}}
\put(0,1){\line(1,0){3}}
\put(0,0){\line(0,1){1}}
\put(3,0){\line(0,1){1}
}
}
\put(10,4.5){\makebox(0,0){$B$}}

\put(12.5,2){\makebox(0,0){$=$}}

\put(18.5,2){\makebox(0,0){$+$}}

\put(24.5,2){\makebox(0,0){$+$}}

\multiput(13.5,0)(6,0){3}{
\put(0,0){\line(1,0){1}}
\put(0,4){\line(1,0){1}}
\put(0,0){\line(0,1){4}}
\put(1,0){\line(0,1){4}
}
}

\multiput(15,3)(6,0){3}{
\put(0,0){\line(1,0){3}}
\put(0,1){\line(1,0){3}}
\put(0,0){\line(0,1){1}}
\put(3,0){\line(0,1){1}
}
}
\end{picture}
}
\end{center}
\NoShow{
\[
\begin{array}{rcl}
\begin{array}[b]{c}
C \\
\overbrace{
\begin{array}{| c |}\hline
 ~~~~~~~~~~ \\
 ~~~~~~~~~~ \\
 ~~~~~~~~~~ \\ \hline
\end{array}
}
\end{array}
&:=& 
\begin{array}[b]{@{}c@{}}
A \\
\overbrace{
\begin{array}{| c | c | c | c |} \hline
~ & ~ & ~ \\
~ & ~ & ~ \\
~ & ~ & ~ \\
\hline
\end{array}
}
\end{array}
~
\begin{array}[b]{@{}c@{}}
B \\
\overbrace{
\begin{array}{| c |}\hline
 ~~~~~~~~~~ \\ \hline
 ~~~~~~~~~~ \\ \hline
 ~~~~~~~~~~ \\ \hline
\end{array}
}
\end{array}
\NoShow{+
\begin{array}[b]{c}
C \\
\overbrace{
\begin{array}{| c |}\hline
 ~~~~~~~~~~ \\
 ~~~~~~~~~~ \\
 ~~~~~~~~~~ \\ \hline
\end{array}
}
\end{array}}
=
\begin{array}{| c | } \hline
~  \\
~  \\
~  \\
\hline
\end{array}
~
\begin{array}{ c }\hline
 \multicolumn{1}{|c|}{~~~~~~~~~~} \\ \hline
 ~~~~~~~~~~ \\ 
 ~~~~~~~~~~ \\ 
\end{array}
+
\begin{array}{| c | } \hline
~  \\
~  \\
~  \\
\hline
\end{array}
~
\begin{array}{ c }\hline
 \multicolumn{1}{|c|}{~~~~~~~~~~} \\ \hline
 ~~~~~~~~~~ \\ 
 ~~~~~~~~~~ \\ 
\end{array}
+
\begin{array}{| c | } \hline
~  \\
~  \\
~  \\
\hline
\end{array}
~
\begin{array}{ c }\hline
 \multicolumn{1}{|c|}{~~~~~~~~~~} \\ \hline
 ~~~~~~~~~~ \\ 
 ~~~~~~~~~~ \\ 
\end{array}
\NoShow{+
\begin{array}{| c |}\hline
 ~~~~~~~~~~ \\
 ~~~~~~~~~~ \\
 ~~~~~~~~~~ \\ \hline
\end{array}}
.
\end{array}
\]
}
    \caption{GEMM performed as a series of rank-k updates.}
    \label{fig:gemmrankk}
\end{figure}
\NoShow{
\devangi{\sout{This is necessary, because now for each of the rank-k updates, the inner size, $ k $, can be chosen to be small, say the $k_C$ encountered in Figure~\ref{fig:BLIS}.}} 
}
This fits  with the discussion in Sections~\ref{sec:CascadingGemm}, where $ k$ had to be constrained to be relatively small. 
This has the added advantage that it limits the workspace required for the cascaded matrices.

A simple approach to implement  FP64x2 \gemm\ via ten FP64 \gemm s is now to explicitly cascade a column panel of FP64x2 matrix $ \widehat A $ into four FP64 matrices, requiring $ 4 \times m \times k_C $ FP64 workspace, and a row panel of $ \widehat B $ into seven FP64 matrices, requiring $7 \times k_C \times n$ FP64 workspace. The ten FP64 \gemm s in (\ref{eqn:Gemm10}) can be computed by making calls to a high-performing DGEMM (FP64 \gemm) implementation like the one provided by BLIS, accumulating these into the four  bins 0, 1, 2, and 3-6 each of size $m\times n$ FP64 numbers. Finally (or after each individual bin is computed), the results of these computations are added into $ C $, via FP64x2 additions, adding in first bin 3-6, then bin 2, bin 1, and lastly bin 0, so that terms that are likely lower precision are accumulated in first.  After this, $ \Sigma $ and $ T $ are applied, scaling the rows and columns of the result.  

The problem with this naive implementation is that it requires considerable additional workspace to accommodate the splits of matrix $A$ and $B$ as well as the bins of matrix $C$.   A careful ordering of the computation can reduce the required space for accumulating the bins. Performance is negatively impacted in this naive approach by the requirement that data be moved between memory layers repetitively as matrices $ B $ and $ A $ are cascaded and $ C $ is assembled from the bins.

\NoShow{Dnp comments: discuss validity of this previous statement with Greg.}

\NoShow{The simplest approach is to explicitly cascade FP64x2 matrix $ A $ into for matrices and $ B $ into seven FP64 matrices, and to then compute the ten FP64 \gemm s in (\ref{eqn:Gemm10}) using a high-performing \gemm\ implementation like the one provided by BLIS, accumulating these into the four bins 0, 1, 2, and 3-6. Finally (or after each individual bin is computed), the results of these computations are added to $ C $, via FP64x2 additions, adding in first bin 3-6, then bin 2, bin 1, and finally bin 0 so that terms that are likely lower precision are accumulated in first.
The problem with this is that it requires considerable workspace ($4 \times m \times k_C$ for matrix $A$, $ k_C \times n \times 7$ for matrix $B$, and $4\times m \times n$ for matrix $C$) and it requires data to move between memory layers as matrices $ B $ and $ C $ are cascaded and $ C $ is assembled from the bins.}


\begin{figure}[tb!]
	\begin{center}
		\includegraphics[width=0.6\textwidth]{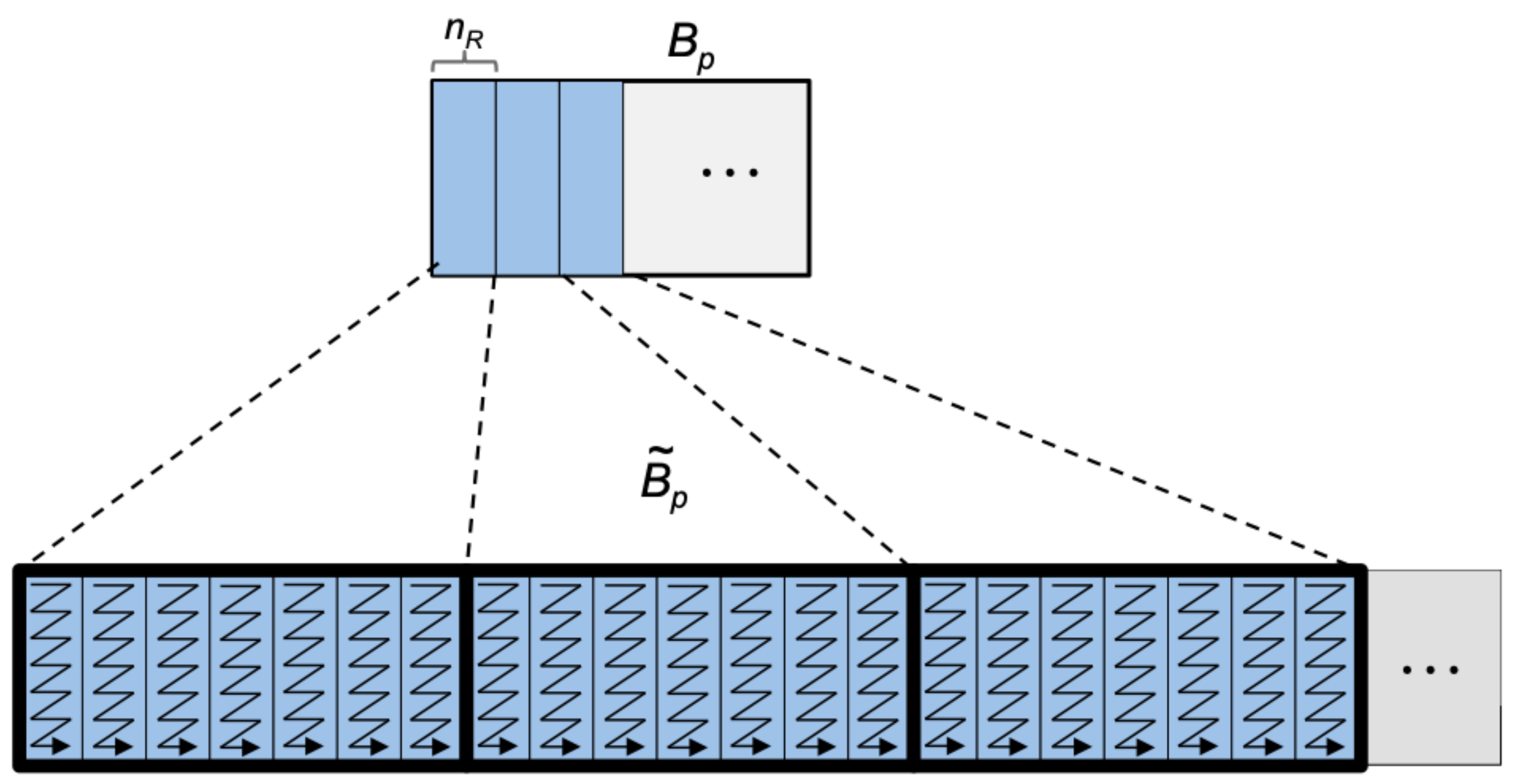}
		\includegraphics[width=0.30\textwidth]{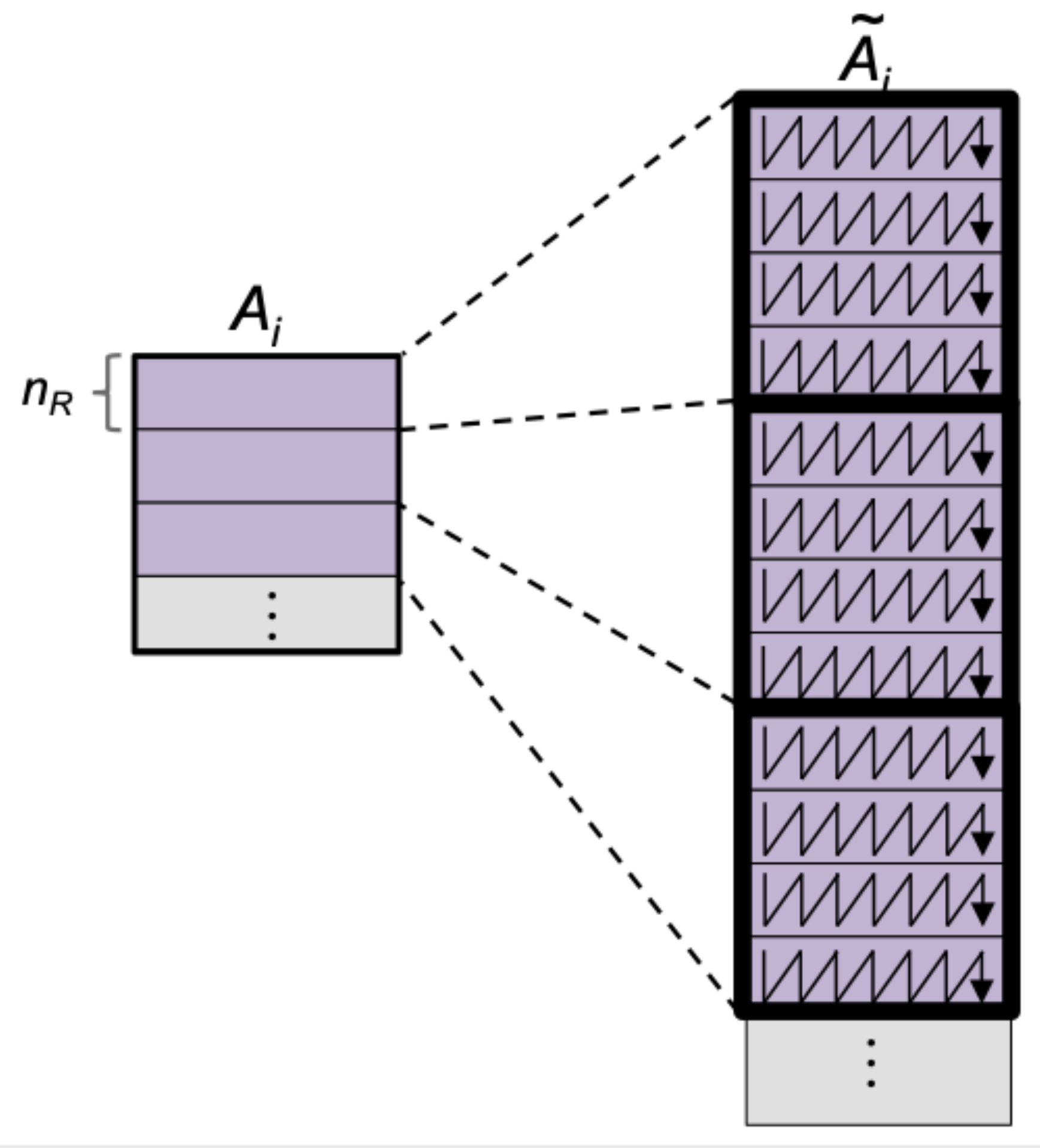}
	\end{center}
	\caption{Left: Packing layout for a panel of $B$. Right: Packing layout for a block of $A$. The colors correspond to the legend in Figure~\ref{fig:BLIS}.}
	\label{fig:packBlayout}
	\label{fig:packAlayout}
\end{figure}

\subsection{BLIS-like: FP64x2 \gemm\ via ten FP64 \gemm s}


\label{sec:BLISCascadingGemm}

\NoShow{
\subsection{BLIS-like: FP64x2 \gemm\ via ten FP64 \gemm s}
}

The kind of overhead incurred by a naive cascaded multiplication is akin to the kind of overhead that makes a naive implementation of Strassen's algorithm~\cite{Strassen} impractical except for very large problem sizes.  In~\cite{Huang:2016:SAR:3014904.3014983} it was shown how integrating Strassen's algorithm into 
an appropriate level of the BLIS \gemm\ algorithm in Figure~\ref{fig:BLIS} it was possible to reduce that overhead so that high performance is achieved even for smaller matrices.
Similar techniques can be employed for the cascaded matrix-multiplication as we now describe. 

\begin{figure}[tb!]
	\begin{center}
		\includegraphics[width=0.65\textwidth]{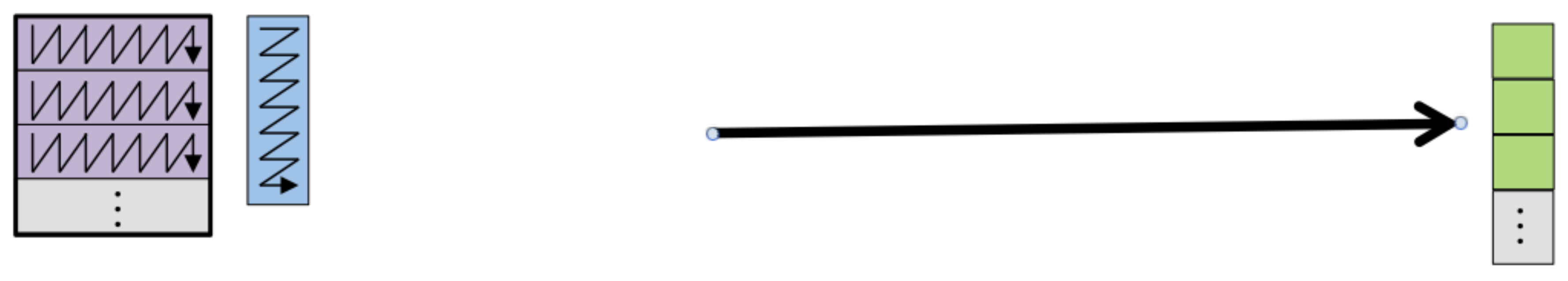} 
    \end{center}
	
    \begin{center}
	\includegraphics[width=0.65\textwidth]{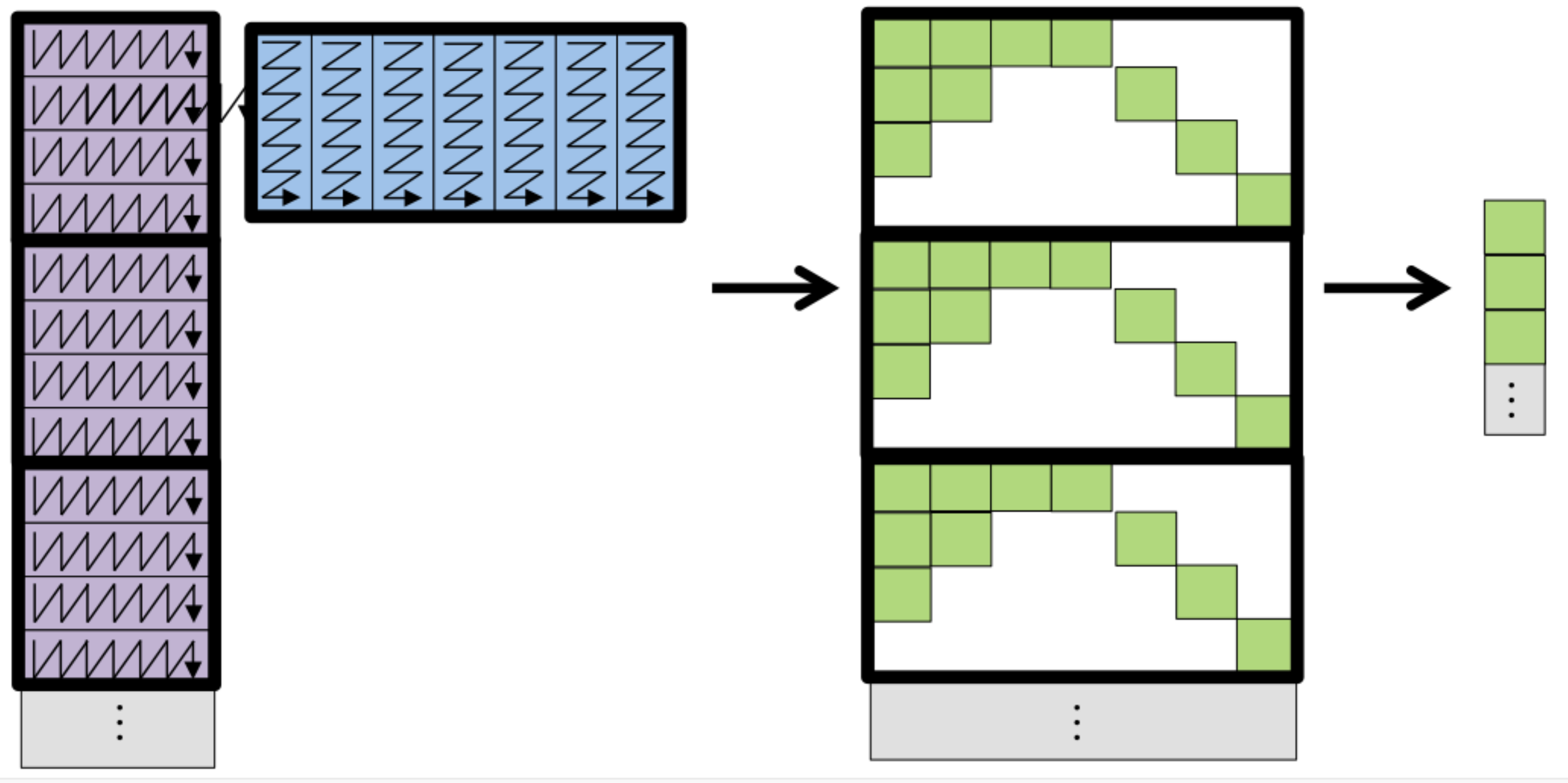}
    \end{center}
    \caption{Top: Operation Goto's Algorithm performs in first loop around the micro-kernel.
	Bottom: Operation performed with cascaded matrices by first loop around the micro-kernel. The colors correspond to the legend in Figure~\ref{fig:BLIS}.}
    \label{fig:computationlayout}
\end{figure}

This approach requires a number of changes to the BLIS framework for \gemm:
\begin{itemize}
    \item 
    During the packing of a row panel of $ B $, FP64x2 micro-panels need to be cascaded into the seven FP64 micro-panels, as illustrated in Figure~\ref{fig:packBlayout}.  This suggests that  the blocksize $ n_C $ be divided by seven so that the packed row panel of $ B $ has the same footprint in the L3 cache%
    \footnote{In practice, $ n_C $ is typically considerably smaller than the L3 cache can accommodate, and hence in practice we divide this by four.}.
    \item 
    During the packing of a block of $ A $, FP64x2 micro-panels need to be cascaded into the four FP64 micro-panels, as illustrated in Figure~\ref{fig:packAlayout}.
    This suggests dividing the blocksize $ m_C $  by four so that the packed block of $ a $ has the same footprint in the L2 cache.
    \item
    The ``packed block of $ A $ times cascaded micro-panel of $ B $'' now for each original micro-panel of $ A $ requires ten calls to the micro-kernel that \response{compute the} micro-tiles.  As part of the execution of the micro-kernel, these contributions can be accumulated into the four bins. These additional calls to the micro-kernel are made by increasing the iterations of the first loop around the micro-kernel by four, and the second loop around the micro-kernel by seven. As indicated by~(\ref{eqn:Gemm10}) all possible products need not be computed.
    \item  The final results is accumulated into the appropriate micro-tile of $ C $, as illustrated in Figure~\ref{fig:computationlayout}, and scaled by applying $ \Sigma $ and $ T $.  This  is done after the first loop around the micro-kernel to ensure the various bins remain in the L2 cache during this computation.
    \end{itemize}
    Importantly, all the described changes lie within the C code that implements \gemm.  Micro-kernels that are already part of BLIS can be reused as-is. These changes do not \response{affect} the threading infrastructure already present in BLIS, enabling multi-threading for this implementation.

\NoShow{
\subsection{Phase three: gluing it back together}
\label{sec:phasethree}
Assuming that during phase two, we have accumulated the results into four bins, we have to add the bins to obtain the final result. To minimize the absorption errors,  we start from the lowest order bin and successively add in the higher order bins. This accumulation is done in FP64x2 arithmetic or with a quick two-sum arithmetic~\cite{}. 
}

\subsection{Adding detection of cancellation errors}
\label{sec:morecancellation}
Detection of cancellation, as hinted at in Section~\ref{sec:detecting}, can be incorporated into the algorithm in the step where a $ m_C \times n_R $ micro-panel of $ \widehat C $ is assembled from the various bins.
This allows cancellation to be detected for multiplications where $ k \leq k_C $.
The current implementation does not attempt to perform any correction if such cancellation is found.

\section{Experimental evaluation}

We provide results in support of the proposed approach by reporting performance and accuracy results.

\subsection{Performance experiments}

\label{sec:perfexp}

We first report preliminary results of the FP64x2 GEMM via ten FP64 GEMMs implementation.  As discussed, the simple implementation proposed in Subsection~\ref{sec:simple} can only become high performing for very large matrices and hence we only give performance results for the BLIS-like implementation.

\NoShow{
One of the problems with high precision is how to measure performance. In some papers, the number of FP64 operations (flops)  are used. This this works well for FPP64 GEMM, which is made up of FP64 dot products, because FP64 dot products have an equal number of adds and multiples. 
{\bf Is this true?:}However, one of the problems in FP64x2 algorithms is that most of the computations are adds and subtracts, and what's more they are chained and sequential operations that depend strongly on the previous result.

In our case, we wanted to simplify our performance story as much as possible. The best comparison then is against DGEMM performance. DGEMM is highly tuned already in BLIS. So if our algorithm runs 10x slower than DGEMM, this has a lot more meaning than just saying it has 10x the flop count. 
}

As is customary for reporting \response{the} performance of \gemm-like operations, we report the rate of computation in billions of FP64 (double precision) floating point operations per second (FLOPS).  
For comparison, the rate at which the BLIS implementation of DGEMM  computes is given, using a $ 2 m n k $ count for multiplying a $ m \times k $ matrix times a $ k \times n $ matrix.
For the FP64x2 cascaded multiply, an operation count of $ 10 \times 2 m n k$ is used. \response{We also compare the performance to the performance of OzBLAS QDGEMM, where QGEMM (FP128 GEMM) is computed using DGEMM,  enabling the ``fast'' mode, which for these matrix sizes and data ranges reduces the number of GEMMs that are computed to 21 yielding an  operation count of $ 21 \times 2 m n k$. We compare against QDGEMM since OzBLAS does not implement FP64x2 GEMM in terms of DGEMM.} In the experiments, all matrices are square.

\NoShow{xs\subsection{Platform}

\label{sec:platform}
}

The performance experiments were conducted on an Intel Core i7-7700K CPU with 4 cores. Each core executes at 4.20~GHz with a max turbo frequency of 4.50~GHz, providing a single-core peak performance of 72~GFLOPS in double precision. The 8~MB L3 cache is shared between all four cores, while each core has a private 256~KB L2 and 32~KB L1 cache. The installed OS is Ubuntu 18.04 running the Linux 4.15.0 kernel. BLIS version 0.8.1 was used in these experiments for the DGEMM performance curves. \response{OzBLAS version 1.5 was built so that it uses DGEMM provided by BLIS.} 

\remove{\bf Should something be said about what version of BLIS was used? Devangi- could you fill this in?}

\NoShow{
\subsection{Performance results}
}

\begin{figure}[tb!]
\begin{center}
\begin{tabular}{@{\hspace{-0.15in}} p{0.52\textwidth}  p{0.52\textwidth}}
\includegraphics[width=0.52\textwidth]{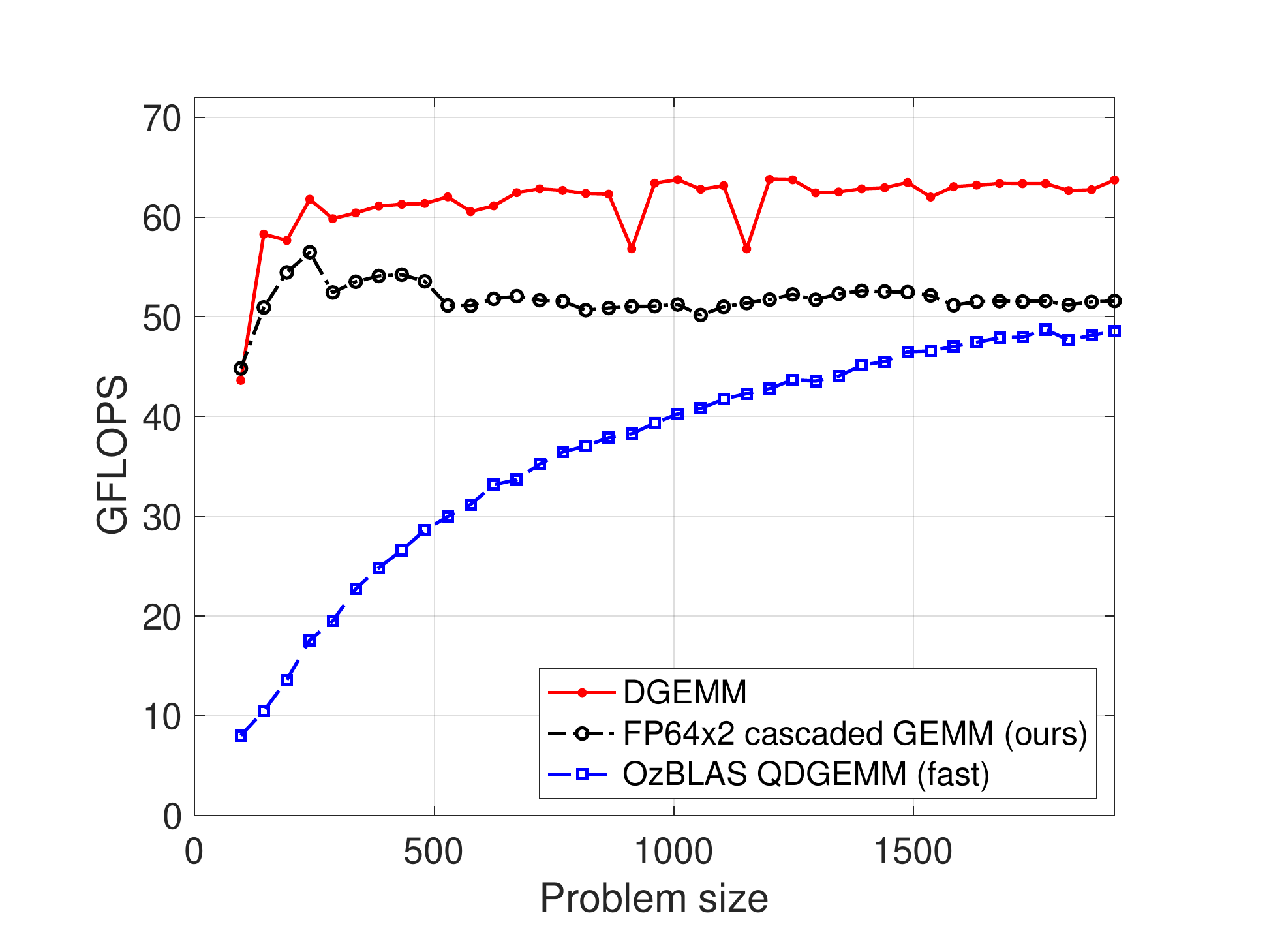}
&
\includegraphics[width=0.52\textwidth]{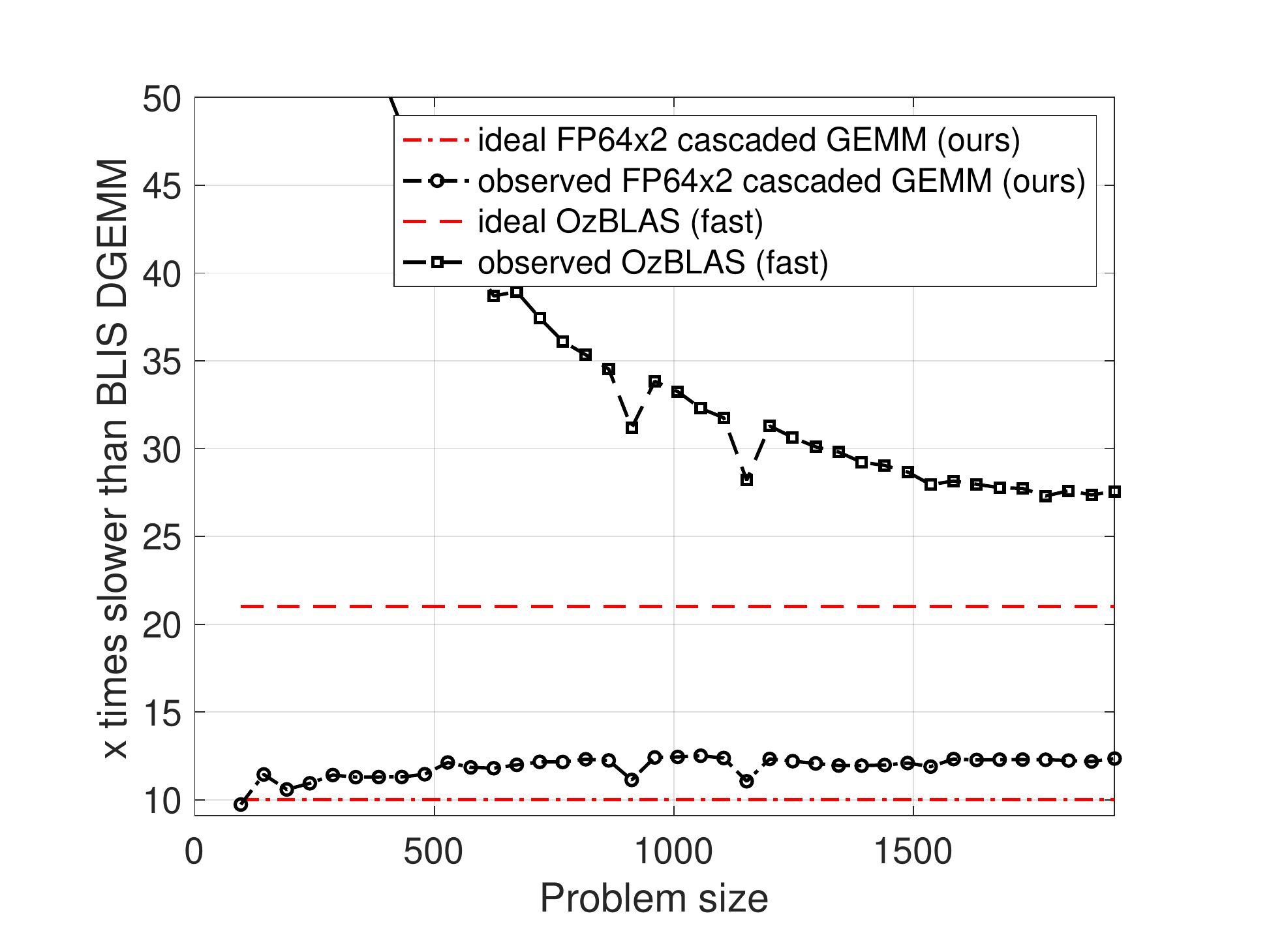}
\end{tabular}
\end{center}
\caption{Comparing the sequential performance of BLIS DGEMM with FP64x2 cascaded matrix implementation. Left: GFLOPS attained for BLIS DGEMM, FP64x2 cascaded \gemm\ at various problem sizes ($m=n=k$) \response{and OzBLAS QDGEMM (fast)}. It is important to realize that for the different curves different operation counts are used.  Right: Ratio of the execution time of FP64x2 cascaded \gemm\ vs BLIS DGEMM (FP64 \gemm) \response{and OzBLAS QDGEMM (fast) vs BLIS DGEMM}. Ideally, we expect $10\times$ slow down for FP64x2 cascaded \gemm\ and \response{$21\times$} slowdown is expected for OzBLAS. The observed OzBLAS slowdown for smaller matrices have not been shown as including them in the graph would have reduced readability.} 
\label{fig:perf}
\end{figure}

Performance results are given in Figure~\ref{fig:perf}.  In the experiments, the cancellation detection mechanism is activated, but no correction is performed if a possible cancellation situation is detected.
In the figure on the left, we report performance as a function of size, for square matrices.  Here, the top of the graph represents peak performance.
The same data is presented in the figure on the right, except that the slowdown relative to DGEMM is given.  Since ten FP64 GEMMs are performed, one would expect to at best observe only a $ 10 \times $ slowdown.  This ideal is not achieved due to the overhead incurred by cascading matrices and summing results. \response{A breakdown of the time spent in each phase is shown in Figure~\ref{fig:breakdown}.}

\begin{figure}[tb!]
\begin{center}
\includegraphics[]{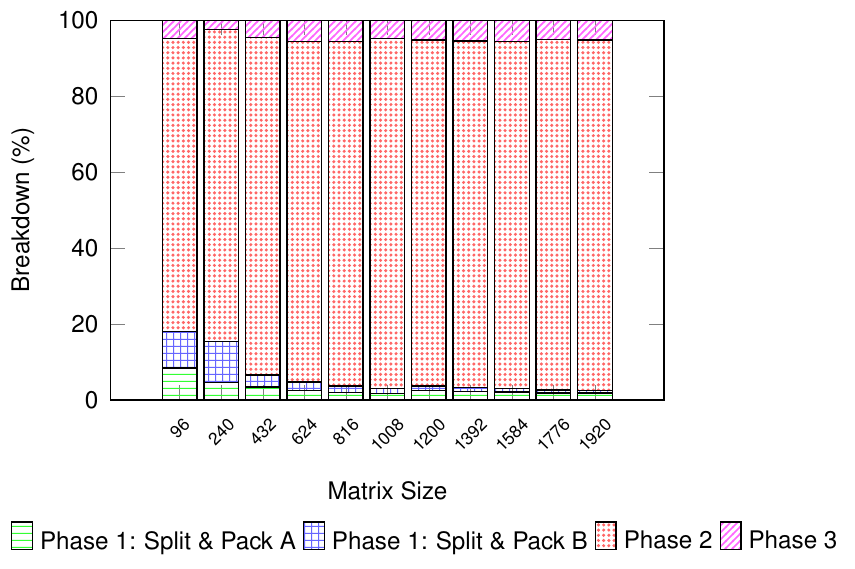}
\end{center}
\caption{Breakdown of the execution time of each phase for FP64x2 Cascaded \gemm.} 
\label{fig:breakdown}
\end{figure}

Our algorithm attains performance in line with  DGEMM efficiency. That is, we incur an overhead due to the cascading of the matrices into its splits (phase one), and the gluing together of matrices in the end (phase three), but it's only slight. The final algorithm does 10 DGEMM calls, and it has 10$ \times $ the flop count, but tends to run 10-13$ \times $ slower than DGEMM. \response{In comparison, OzBLAS, for most of the matrices we tested, performs 21 \gemm s in the fastmode but tends to run approximately 27$\times$ slower than DGEMM. The OzBLAS implementation is not only more expensive due to the number of \gemm s performance, but also ramps up to high performance slowly because calling separate \gemm s increases the incurred memory traffic, which represents overhead, amplified by 
less efficient implementation of their phase 1 and phase 3.}
\response{In summary, OzBLAS is more than twice slower than our method, even when using the fast mode approach, and even worse when the fast mode approach is not enabled. It also scales worse, is less efficient, and ramps up to high performance more slowly.}

\response{We do not compare our performance with a FP64x2 and FP128 implementation.  The primary reasons are that we do not have a competitive x86 implementation of either. We found a few FP64x2 implementations that were actually triply-nested loops and had no register or cache-blocking. The situation with FP128 is even worse since it's based entirely on software emulation.} Moreover, any implementation would be subject to questions of whether it was the highest performing that can be implemented, regardless of the attained performance. \response{We instead circumvent this by directly comparing to DGEMM performance, whose implementation is already well studied and understood.}
\NoShow{What's more is that it has the same threading opportunities as BLIS's DGEMM and it fits into the current BLIS infrastructure. 
}

\NoShow{For our performance runs, we enabled the detection of cancellation discussed in Subsection~\ref{sec:detecting}, but not the overhead of actually correcting any detected problem cases. 
}

\subsection{Accuracy experiments}\label{sec:accuracy}

A question is how accuracy is affected by casting FP64x2 \gemm\ in terms of ten FP64 \gemm s.  We here address this empirically,
by investigating two types of problems: \gemm\ with matrices filled with uniformly distributed random data in a given range (or, sometimes widely varying ranges) and matrices constructed so that \gemm\ is ill-conditioned.  The experiments examine how the error experienced by our scheme compares to the error encountered by a true FP64x2 multiplication. We use FP128x2 (double-quad) for reference accuracy.

\subsubsection*{Uniformly random data experiments}

For our experiments, we ran a series of tests over problems with points initially (pseudo-random uniform generator) picked in the given range. Uniformly random data doesn't tend to spot all irregularities associated with range, which is why we also added ill-conditioned and wide-ranging experiments. Unless otherwise stated, we ran the uniformly random experiments on $[-1,1]$. We first picked a random FP64 number in the given range, and then added random bits in the mantissa past 53, up to 113, and used this as a FP128 number in the prescribed range. We then converted this FP128 number into non-overlapping double-double FP64x2 format. This typically means that even though our exponents are non-overlapping between the high and low part, that they don't have too wide of a gap (only up to 8 or so bits in the initial exponent.) However, this seemed good enough, since larger gaps would probably get lost in the noise of the experiment.

\subsubsection*{Wide-range data experiments}

For this set of experiments, instead of picking numbers in a range randomly, we picked the exponent ranges themselves randomly, and then picked uniformly random \response{mantissa} within those ranges. Using this method, it is possible to have some dot products cover ranges like $[2^{-30},2^{0}], [-2^{10},2^{3}]$, or any combination. For each row of matrix $A$ and each column of matrix $B$, we first randomly picked exponents between $10^{-60}$ to $10^{20}$, and next randomly assign a sign. This gives us a range $[x,y]$, from which we pick uniformly random numbers to fill in the the rows of matrix $A$ and columns of $B$. This results in having as many possible different dot products, with different ranges on the inputs, as possible.

\remove{ \subsubsection*{Wide data range experiments}
For this set of experiments, instead of picking numbers in a range randomly, we picked the exponent ranges themselves randomly, and then picked uniformly random mantissas within those ranges afterward. This way, it's possible to have some dot products cover ranges like $[2^{-30},2^{0}], [-2^{10},2^{3}]$, or any combination. To generate this data, we first picked random exponents given from $10^{-60}$ to $10^{20}$. We picked random signs next. Then, once $[x,y]$ is randomly determined, we find uniformly random data in those points. We do this independently for every row of $A$ and column of $B$, so in the end, we have as many possible different dot products, with different ranges on the inputs, as possible. Any method that treats all of $A$ or $B$ the same needs to be tested in this way, otherwise, one might not realize weaknesses in the method.}

\subsubsection*{Ill-conditioned experiments}

One would expect error to be worst when cancellation is encountered.  For this reason, since matrix multiplication in effect computes an entry of $C = A B $ by taking the dot product of a row of $ A $ with a column of $ B $, it pays to examine when a dot product of two non-zero vectors, $ x $ and $ y $, with reasonable length yields a relatively small result compared to their initial norms. \response{The condition number of a non-zero dot product is given by
\begin{equation}
    {\rm  condition} = \frac{\| x \|_2\| y\|_2}{| x^T y |}. 
    \label{eq:condnum}
\end{equation}}
This ill-conditioning happens when, for example, $ \| x \|_2 \approx \| y \|_2 \approx 1 $ and $ x^T y $ is small.  The smaller $ x^T y $, the more likely the computation will have cancellation errors.

An attempt to create matrices so that multiplication with them is ill-conditioned was proposed when testing XBLAS~\cite{xblasweb}. They wrote an excellent generator for generating two vectors whose dot product manages to hit major cancellation error, and therefore, to generate matrices for matrix multiplication, one could stuff generated vectors $x$ into rows of $A$ and generated vectors $y$ into columns of $B$. While effective, this approach leads to problems just on the diagonal of the result matrix, and meaningless (potentially very well-conditioned) data elsewhere. It seems unfortunate to generate $n^2$ data but \response{have only $n$ elements of them be useful for testing purposes.} We propose an approach that yields many more potentially ill-conditioned cases within each \gemm, thus leveraging the work to compute \gemm\ towards many more cases of interest. 

We create matrices that incorporate a large number of dot products that induce cancellation.  For now, assume that the matrices are all square ($ m = n = k $).  
We start by generating (in high precision) matrix $ A $ to have mutually orthonormal rows, by generating a random matrix, computing its FP128 QR factorization (using quad Householder transformations), and setting $ A $ equal to the resulting $ Q $.
Finally, we generate a matrix $  C $ by first choosing its elements to be small, with magnitude in the range $ (t, 10 t ) $, where $ t > 0 $ is some small positive tolerance. We randomly pick these values to be positive or negative (that is, each final element of $ C $ is either in $(t,10 t )$ or $(-10t, -t )$.) We also randomly pick precisely one location in every column of $ C $ to equal 1. 
Next, we set $ B = Q^T  C $ in quad-arithmetic.
    \response{Now,
each row, $ e_i^T A $,  of $ A = Q $ has norm equal to one and,  
if $ t $ is small, each column, $ B e_j $, of $ B $ equals 
approximately one (since $ B e_j = Q^T C e_j $ and multiplication by a unitary matrix preserves length).
    Thus, since the $ i,j $ element of $ C $ equals $ e_i^T C e_j = (e_i^T A) (B e_j) $ which, for most choices of $ i$ and $ j $, will have a magnitude bounded by $ 10 t $,  taking the product of $ A $ and $ B $ will incur mostly dot products with high condition numbers, if $ t $ is chosen to be small, triggering cancellation.
}

The above description works with square matrices and may be modified for non-square matrices. 
Since the accuracy of finding the worst component-wise error is independent of the matrix dimensions, we only examined results from multiplying square ill-conditioned tests. 
Of interest now is how elements in the result matrix lose accuracy as measured by    
the component-wise relative error.  

\begin{figure}[tb!]
\begin{center}
\begin{tabular}{@{\hspace{-0.15in}} p{0.52\textwidth}  p{0.52\textwidth}}
\includegraphics[width=0.52\textwidth]{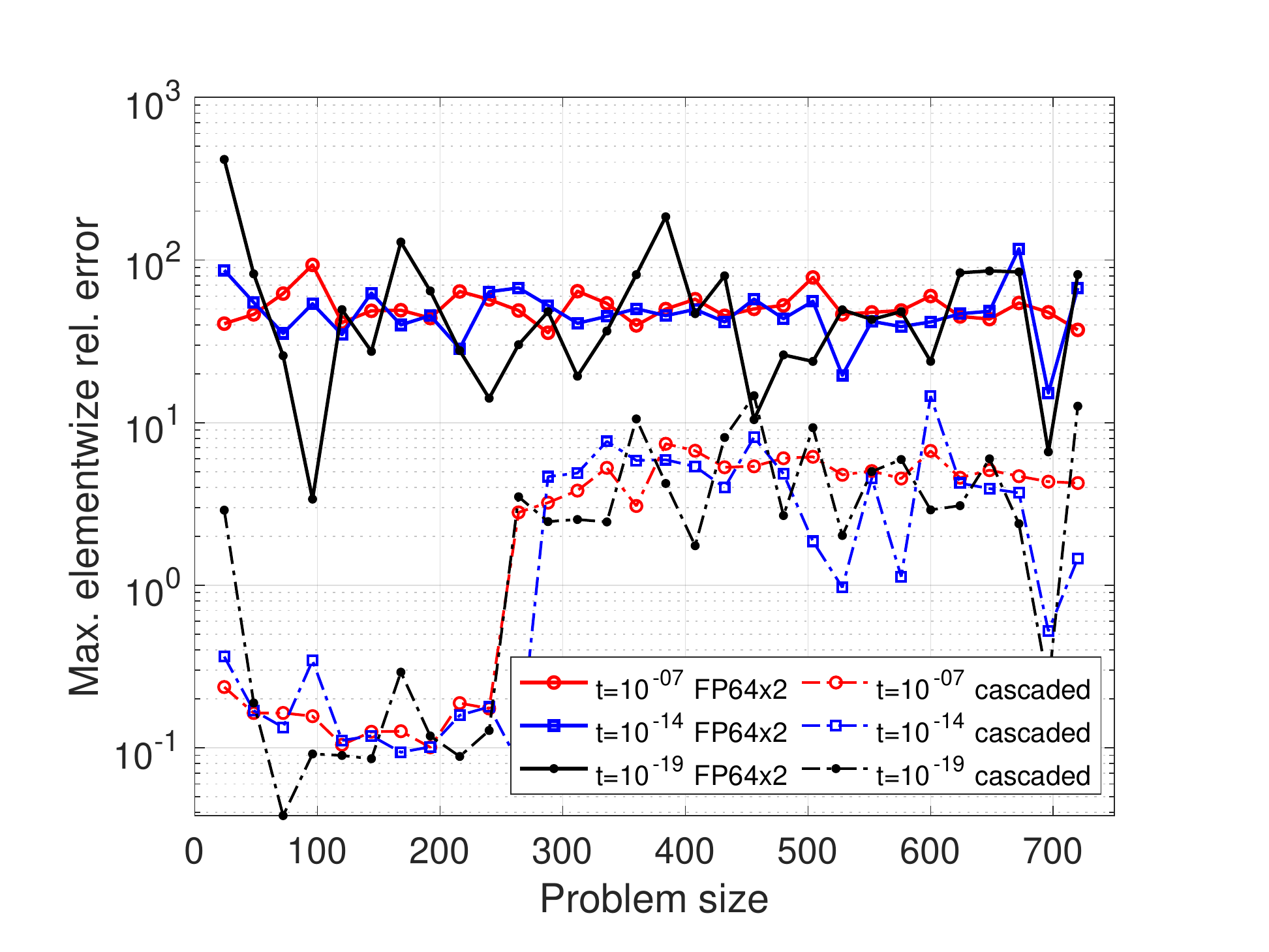}
&
\includegraphics[width=0.52\textwidth]{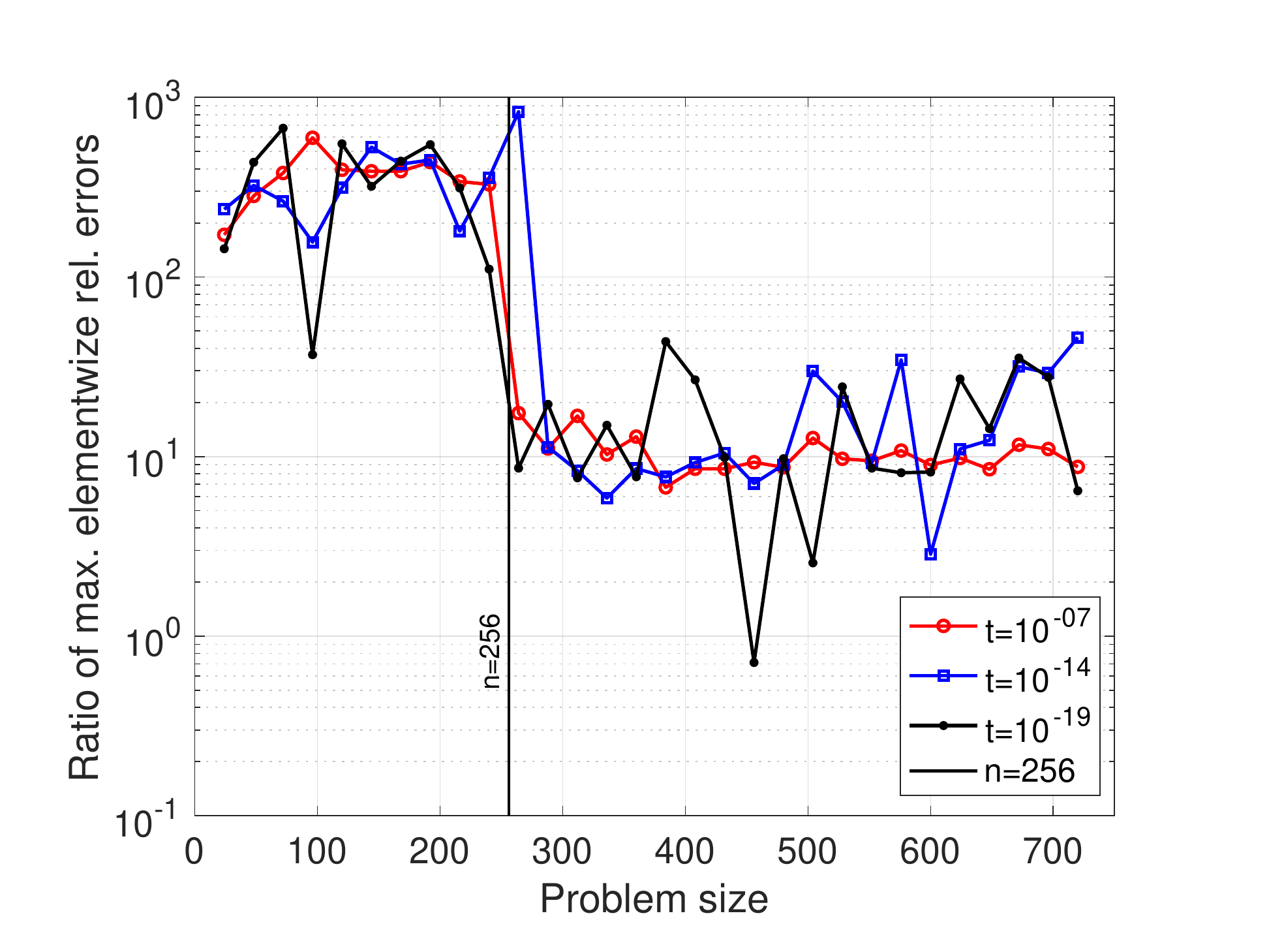}
\end{tabular}
\end{center}
\caption{Error encountered for problems constructed to be ill-conditioned. Left:  Maximum component-wise error for various choices of tolerance, for computation with FP64x2 arithmetic and with cascaded matrices. 
Right: Same data, presented as the ratio of error when computing with FP64x2 arithmetic and with cascaded multiplication.
\NoShow{\bf Fix ``data3''??}}

\label{fig:illcond}
\end{figure}

In our experiments, we focus on three choices for tolerance $ t $:
\begin{itemize}
    \item 
    $ t = 10^{-9} $, which can wipe out all accuracy when computing in single precision (FP32).
    \item 
    $ t = 10^{-14} $, which can wipe out all accuracy when computing in double precision (FP64).
    \item 
    $ t = 10^{-19} $, which can lead to significant loss of accuracy when computing in FP64x2 arithmetic, but will leave some accuracy.
\end{itemize}
For each of these choices, we ran experiments for a range of problem sizes, \response{reported} 
in Figure~\ref{fig:illcond} the average worst case  component-wise errors over many runs with different generated matrices.
The error is computed as the difference with the corresponding value in the result matrix when computed in
very high (FP128x2) arithmetic.
We note that by this measure, the accuracy of the cascaded multiply in these experiments is vastly superior to that of FP64x2.

In Figure~\ref{fig:illcond} (left), we report the maximal elementwise relative error when computing \gemm\ with a simple triple-nested loop in FP64x2 arithmetic and with our cascaded multiplication, for different choices of tolerance $ t $.  The same data is reported in Figure~\ref{fig:illcond} (right), which reports the ratio between the maximal element-wise relative errors encountered \response{by} the FP64x2 implementation and the cascaded multiplication.  
What we see is that our approach tends to be {\em more accurate} for the ill-conditioned experiments. 
We attribute this to the fact that the highest order bits are computed error-free regardless of the condition number.
It is an encouraging result that wide-ranging data seems to behave as well as the well-conditioned data: it means that our methods handle range issues well, in contrast to most algorithms that employ fixed-point. 

The drop in the ratio \response{of maximum element-wise relative errors at} size 256 can be expected: this is where, in the fourth loop around the micro-kernel, BLIS starts performing multiple rank-k updates. This means that each element in the result matrix is updated multiple times, incurring multiple conversions to FP64x2 storage.  The point is that once we rounded down to FP64x2 accuracy, the benefits of storing intermediate results in higher precision diminishes. 
\NoShow{
{\bf What are you trying to say here?}  That is, what one sees is not so much a loss of accuracy in the algorithm, but the loss of accuracy just involved in considering FP64x2.
}

A question is how accuracy is affected on an element-by-element basis.  To answer this, we examine the relative error in all entries of matrix that results when multiplying two $ 240 \times 240 $ matrices. 
In Figure~\ref{fig:illcond2} we report the ratio between the error encountered by the FP64x2 implementation and the cascaded matrix multiplication for different cases: 
\begin{itemize}
\item
Uniformly random data in $[-1,1]$.
\item
\response{Wide-ranging} data. 
\NoShow{\devangi{Moved the paragraph explaining how these experiments earlier in the section.}
}
\remove{\robert{Instead of picking numbers in a range randomly, we picked the exponent ranges themselves randomly, and then picked uniformly random mantissas within those ranges afterward. This way, it's possible to have some dot products cover ranges like $[2^{-30},2^{0}], [-2^{10},2^{3}]$, or any combination. We first picked random exponents given from $10^{-60}$ to $10^{20}$. We picked random signs next. Then, once $[a,b]$ is randomly determined, we find uniformly random data in those points. We do this independently for every row of A and column of B, so in the end, we have as many possible different dot products, with different ranges on the inputs, as possible. Any method that treats all of $A$ or $B$ the same needs to be tested in this way, otherwise, one might not realize weaknesses in the method.}}
\item
Very badly conditioned: $ t = 10^{-19} $.

\end{itemize}
Except for a very small number of elements, the cascaded multiplication yields the same or better accuracy.

In Figure~\ref{fig:illcond4}, we again present the relative error in all elements of the result matrix, but this time sort the results by the error occurred by the cascaded multiplication.
In other words, for a given element (value along the x-axis), the error incurred by the cascaded multiplication is reported in red and the error incurred by the FP64x2 implementation in black.  This shows how for well-conditioned problem\response{s} the cascaded matrix multiplication is essentially uniformly better while for the ill-conditioned and wide-ranging inputs only in a few cases the FP64x2 implementation is better.

\NoShow{
{\bf Comment:  Shouldn't the FP64x2 also be implemented using the BLIS framework, in order to mimic a more similar order of updates to elements?}
{\bf Response: While I'd love to know how fast FP64x2 can be made possible, it doesn't seem that this strategy would ever be faster than the 10-GEMM method. So the ONLY purpose for this code would be for this paper. In which case, does this paper actually need this experiment to be publishable?}
}

\begin{figure}[tb!]
\begin{center}
\includegraphics[width=0.52\textwidth]{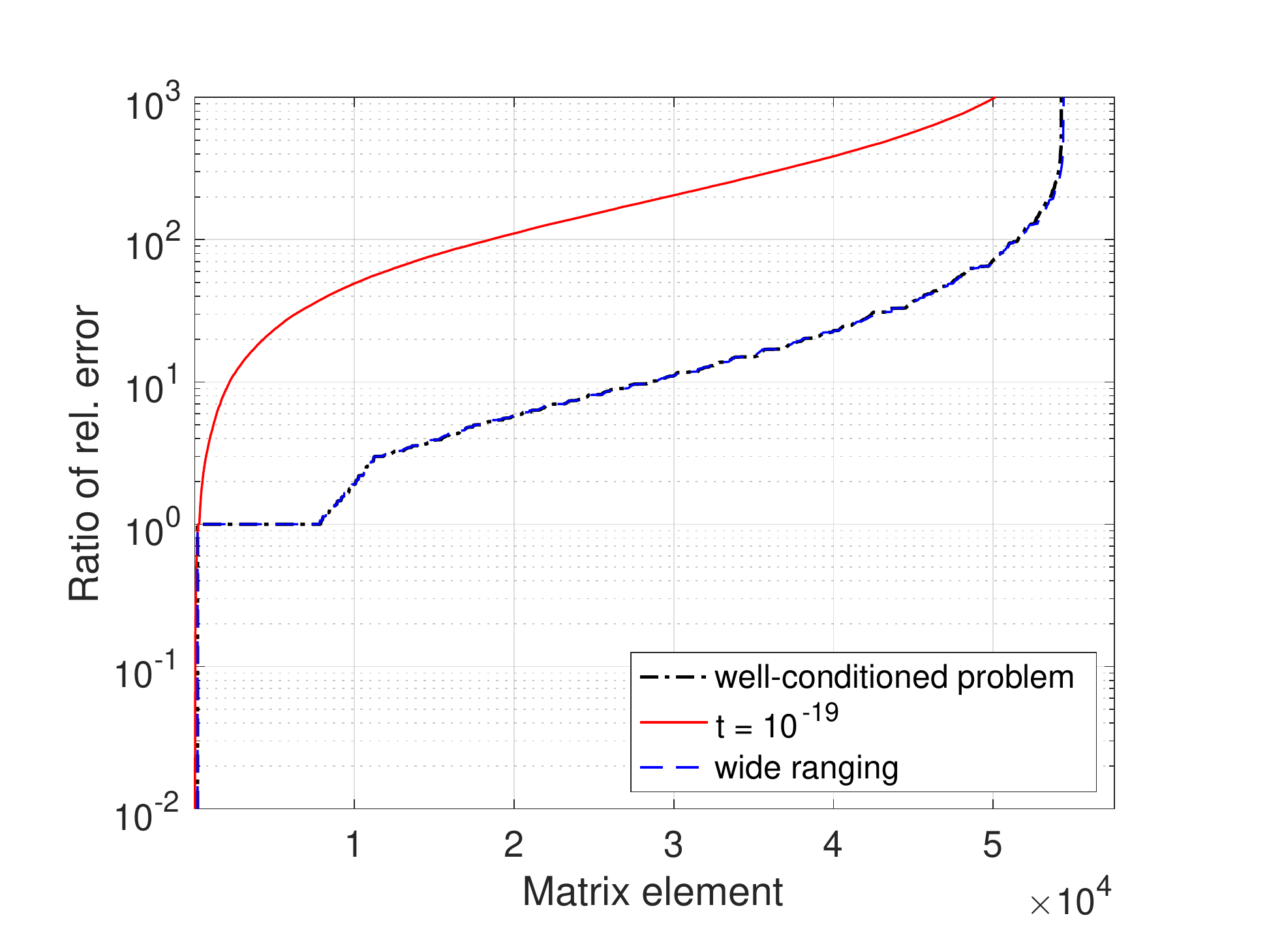}
\label{fig:illcond3}
\end{center}

\caption{
For this graph, the ratios (Cascaded divided by FP64x2) of relative accuracies of all elements of  matrix $ C = A B $ are reported, where all matrices involved are $ 240 \times 240 $.
The elements are sorted by the value of the ratio, independently for each curve.  Values (towards the left of the graph) that are less than $ 10^0 $ mean that FP64x2 is more accurate. The well-conditioned \response{(uniformly-distributed data where the input data lies within the range of [-1,1])} and wide ranging values experiments are virtually indistinguishable in the graph. 
}

\label{fig:illcond2}
\end{figure}

\begin{figure}[tb!]
\begin{center}
\begin{tabular}{@{\hspace{-0.15in}} p{0.52\textwidth}  p{0.52\textwidth}}
\includegraphics[width=0.52\textwidth]{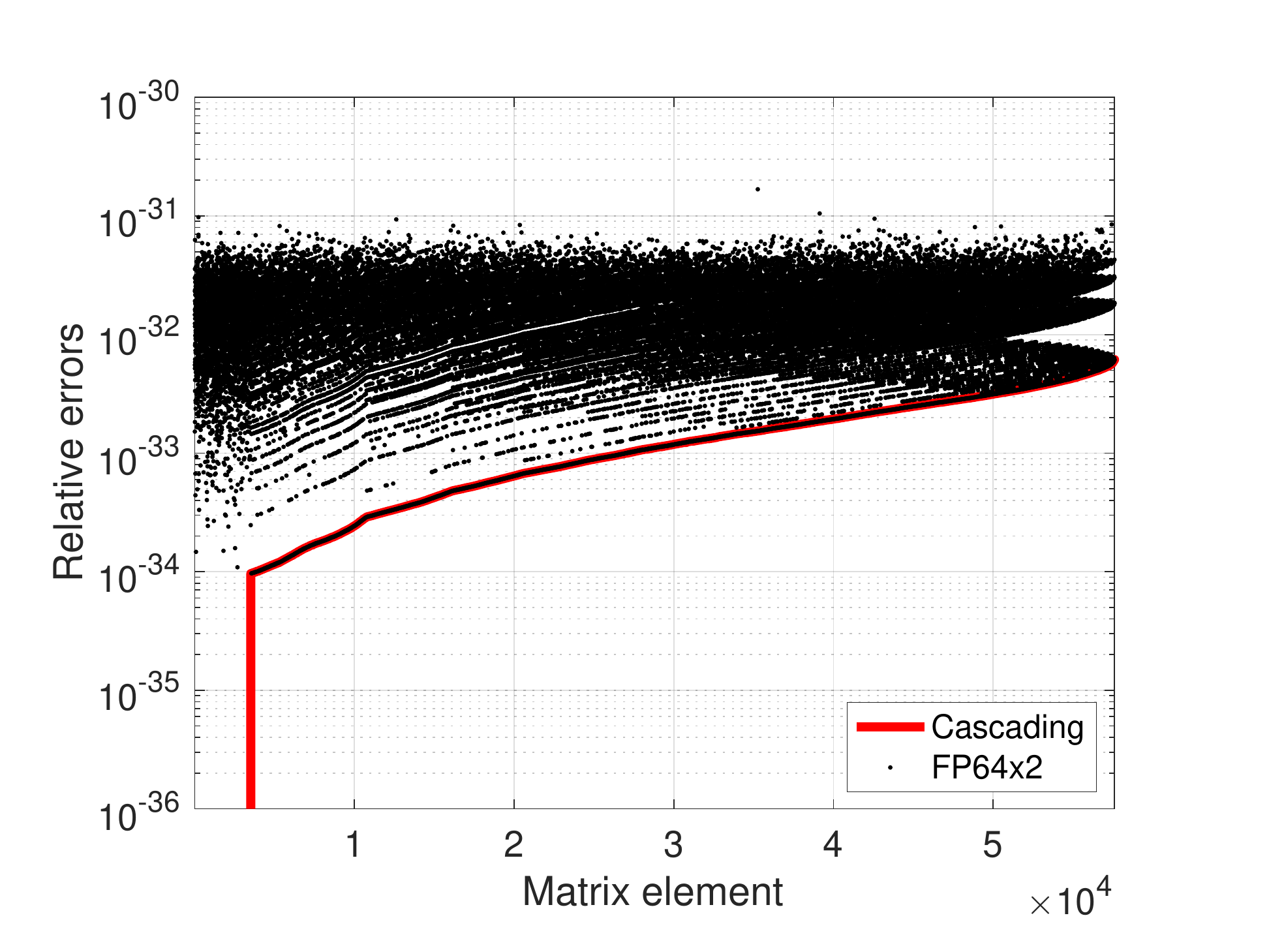}
&
\includegraphics[width=0.52\textwidth]{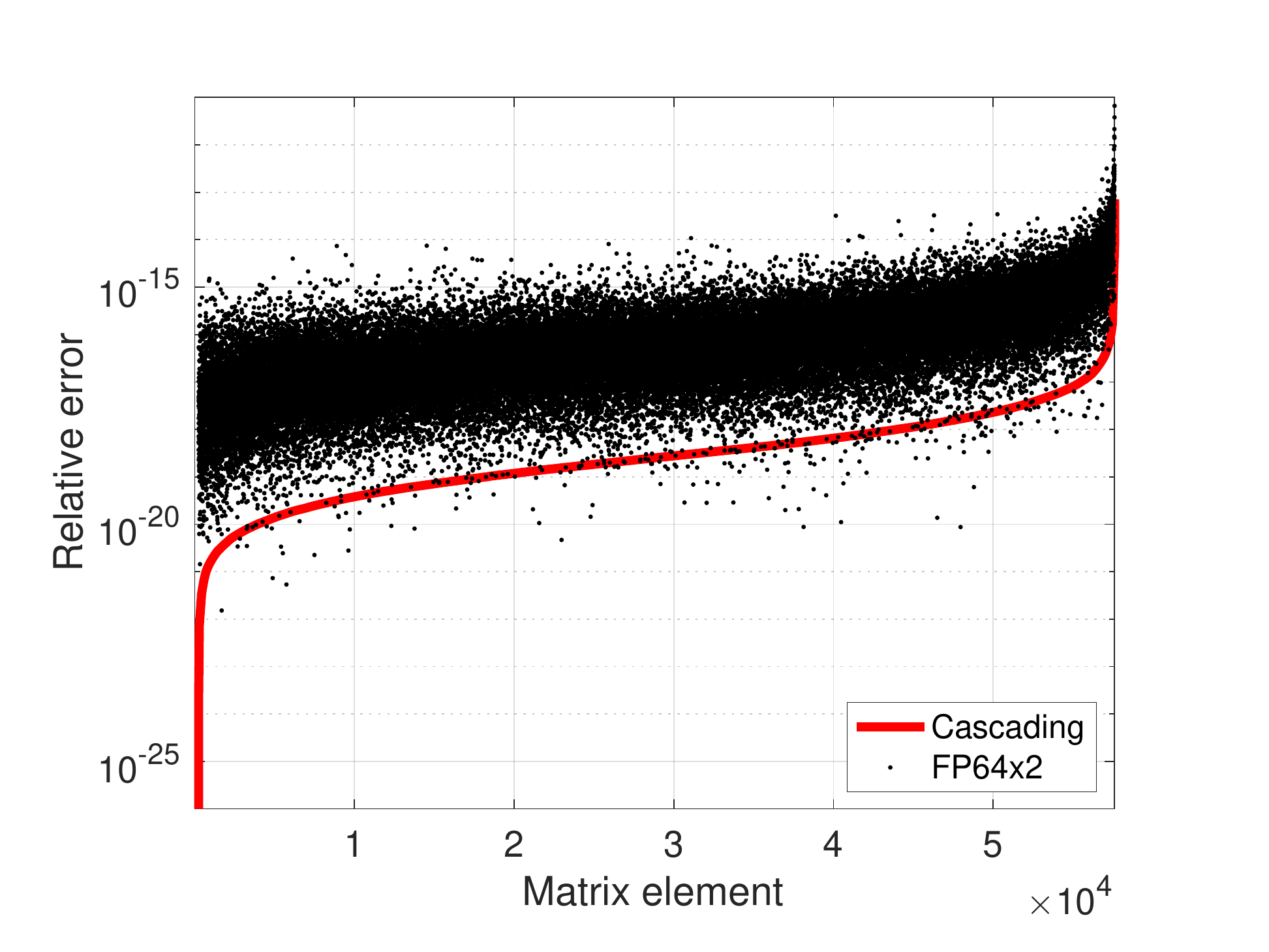}
\end{tabular}
\end{center}
\caption{
The relative error of all elements of  matrix $ C = A B $ are reported, where all matrices involved are $ 240 \times 240 $.
The elements are sorted by the relative error incurred by the cascading matrix multiplication.  The corresponding error incurred by FP64x2 is also reported.  Left: well-conditioned experiments \response{(uniformly-distributed data where the input data lies within the range of [-1,1])} .  Right: experiments with conditioning in the $ \response{10^{-19}} $ range.  Be aware of the different scale reported along the y-axis.
}

\label{fig:illcond4}
\end{figure}

\section{Conclusion and future opportunities}

We reiterate that our focus on casting FP64x2 in terms of FP64 is merely a concrete illustration of a general principle.  We choose to target bootstrapping FP64x2 accuracy from a high-performance FP64 implementation because BLIS as of the writing of this paper only supports double (FP64) and single (FP32) precision.
By leveraging BLIS, we have demonstrated that high performance and accuracy can be achieved.

We now discuss a few extensions of the presented ideas that can be pursued in the future.
 We believe that these constitute several theses  worth of work and can hence not be expected to be addressed in this initial study.

\subsection{Safely dropping lower order bins}

\remove{{\bf Greg:  Can you tie this to something that is in the paper?  Or can we say something really vague here?}

\robert{I don't get this...}

Consider the case where matrix $A$ and $B$ is split into $y$ chunks, where each chunk holds $c$ bits. This means $cy$ bits in the input matrices are tracked. To compute $C=AB + C$, with at least $cy$ bits, using the proposed cascading \gemm\ methodology, we must compute at most $y^2$ multiplies and at least $\frac{y(y+1)}{2}$ multiplies.  While generalizing the proposed method, this helps  determine how many lower precision \gemm s need to be performed to compute the higher precision output. In the case study, presented in this paper, this equation also indicates the number of error-free multiplies. 

When doing any integer quantization algorithm, usually the $\frac{y(y+1)}{2}$ formula suffices. To follow up with FP64 arithmetic, we need an additional $y$ multiplies.
}
Generalize Equation (\ref{eq:dot}) where each vector is split into $t$ splits and each split has a similar bit-width.  Notice that {\em if} bin~0 has few or no leading zeroes, {\em then} any contributions from lower order bins (in Equation (\ref{eq:dot}) bins~4--6) may not affect the final stored result.  This observation illustrates that under the above conditions not all terms encountered in Equation~(\ref{eq:dot}) or Equation~(\ref{eqn:Gemm16}) may need to be computed.  
Roughly, if the highest order bin has few leading zero bits, only $ t ( t+1)/2 $ terms may need to be computed.  If the highest order bin has many zeroes (or is entirely zero) or the splits are of completely different bit widths then additional terms may be required. 

However, in our case where the first split might be 22 bits and the last 53 bits, and the lower-order terms are necessary. 
Our proposed methods reduces the number of \gemm s by combining the lower-order splits. 
In contrast, if all the splits had been similar, around 21 bits in size, such as done in the OzBLAS, but constrained $t$ to be 4, then we could have done 10 DGEMMs just by dropping the lower order terms also drastically reducing the accuracy of the result. Our method of combining the last few terms is superior but keeps us at 10 DGEMMs. 

\remove{In the event that all the bins hold an equal number of bits, it is useful to quantize how many matrix multiplies occur. As mentioned in section~\ref{sec:generalizing}, if there are $Y$ bins that are all equal in size, one must do at most $Y^2$ multiplies, but at least $Y*(Y+1)/2$ multiplies. For BF16x3, we need 6 of the 9 potential multiplies to approach FP32 accuracy \cite{HenHeiTan2019}.

For us, the number of error-free multiplies also uses this same equation. So for four cascading splits like we use for FP64x2 approximation, three of the bins must be done error free, and that's 6 multiplies. The other four involve some of the $A3$ or $B3$ products. 

When doing any integer quantization algorithm, usually the $Y*(Y+1)/2$ formula suffices. To follow up with FP64 arithmetic, we need an additional $Y$ multiplies. So a Int8x7 algorithm that is only fixed point will need a minimum of 28 multiplies.}

\subsection{General casting of high precision in terms of low precision} 

\label{sec:more_generalizing}

Based on preliminary analysis, we briefly discuss other cases where the general principle of cascading of matrices and multiplication may be of value. While the cases listed here are not exhaustive, one can extended these ideas to other cases that may be applicable. 

\subsubsection{FP32 in terms of bfloat16}
\label{sec:fp32fp16}

Of particular interest these days is how to exploit various ``half-precision'' arithmetic (like bfloat16) that is starting to be supported in hardware~\cite{bfloat16}.  Often, such bfloat16 arithmetic is 8-32 times faster than single precision (FP32) arithmetic%
\footnote{Many times these operations can be memory-bound, and we do not see the full theoretical speed-up}.

The proposed technique can be modified to accommodate cascading FP32 \gemm\ in terms of bfloat16 \gemm. In this case, FP32 numbers can be cascaded in terms of three bfloat16 splits, with appropriately chosen $ D_0 $, $ D_1 $, and $ D_2 $ (as a function of $ k $), and the cascaded multiplication can be cast in as few as six  bfloat16 \gemm s~\cite{HenHeiTan2019}. The FP32 bit accumulation has a number of benefits:  the bins do not first need to be accumulated separately, which simplifies the implementation, reduces the amount of workspace needed and changes how registers and caches are leveraged.

\remove{Our techniques can be modified to accommodate cascading FP32 matrix multiplication in terms of FP16.  
This time, FP32 numbers can be cascaded in terms of three chunks, with appropriately chosen $ D_0 $, $ D_1 $, and $ D_2 $ (as a function of $ k $), and the cascaded multiplication can be cast in as few as \devangi{6}  FP16 \gemm s.
What is particularly interesting is that 
such hardware often supports computing in FP16 while accumulating in FP32.  This has a number of benefits:  the bins do not first need to be accumulated separately, which simplifies the implementation and changes how  registers and caches are leveraged.}

\subsubsection{FP64 in terms of FP32}
This case does not fit the cascading \gemm\ methodology well. The primary problem is that FP64 is supported in hardware on most current architectures.  A secondary problem is that the number of splits becomes unmanageable (we would need approximately five FP32 splits with $ k = 256$ in order to make the first few bins error-free) and hence the number of FP32 \gemm s becomes impractical to gain any speedup.

\remove{This case appears to not fit the cascading \gemm\ idea well.
The primary problem is that FP64 is supported in hardware on most current architectures.  A secondary problem is that the number of chunks becomes unmanageable (7 to 8 chunks with $ k = 256 $) and hence the number of FP32 \gemm s becomes impractical.  Hence, unless architecture design changes the ratio between the speeds of FP64 and FP32 arithmetic dramatically, there is no point in pursuing this.}

\subsubsection{FP64x3 in terms of FP64}
Just as this paper shows how to use a cascading GEMM to approximate FP64x2 accuracy, one can approximate FP64x3 \gemm, by cascading the FP64x3 matrix into 6 FP64 splits. 

\subsubsection{FP64 in terms of INT8}
Modern architectures often support INT8 arithmetic with 32-bit accumulation. This lends itself to the cascading \gemm\ methodology, since to cascade the matrix into its splits, the matrix is converted into fixed point numbers which can be stored in integers. 

To compute FP64 \gemm\ by cascading into INT8, we need at least seven INT8 splits to hold all the mantissa bits of an FP64 number (dropped bits can still occur in this method.) With seven splits, we need at least 21 multiplies to obtain the result. The 32-bit accumulation provides the same benefits of as mentioned in section~\ref{sec:fp32fp16}. A similar strategy has been proposed in~\cite{ootomo2024dgemm}.

\subsubsection{FP64x2 in terms of INT8}
To support FP64x2 \gemm\ by cascading it into INT8, we need approximately 14 to 15 INT8 splits and about 105 to 120 INT8 multiplies to obtain the result. On specific processors, INT8 may run faster than DGEMM. According to Figure~\ref{fig:perf}, we can obtain an FP64x2 \gemm\ result in approximately 12x more time than DGEMM. Therefore, on any processor, where the peak INT8 operations runs at least $120/12 = 10$ times faster than double precision floating point operations, we should see a net speedup using 15 cascading INT8 splits using the methodology developed here. In fact, on many processors, INT8 \gemm\ may run 30x faster than DGEMM, making an INT8 version of cascading matrices potentially three times faster than the results we have presented. 

The reason we do not present FP64x2 \gemm\  by cascading it into INT8 in this work is that we wanted to leverage existing kernels in BLIS, which meant basing an algorithm on DGEMM or SGEMM kernels. {However,} cascading matrices, in general, covers a wide range of problems not yet considered and further justifies the usefulness of this contribution. 

\remove{
{\bf I don't understand any of this.  If $ k = 256 $ then we need around $ \log_2(k) = 8 $ bits to account for accumulating...  that seems to completely ruin the opportunity to use an Int 8 number.}
\greg{ You are forgetting that Int8 usually has a 32-bit integer accumulator (so I added a new paragraph below so others won't stumble here- please let me know if you find it sufficient. So even when $k=256$, we can add up a bunch of Int8s and never worry about overflow. In fact, the range to prevent overflow is even wider than our high precision case. It's the opposite of your worry - this case is actually easier.}

With FP64, we concerned ourselves with getting exact answers in the first few bins, and blocking to keep our integers less than $2^{53}$ to be exact. We assume that Int8 has a wider-accumulator, say Int32. Now we just have to ensure that nothing in our most important first bins can overflow Int32. 

Everything is exactly as we have, just potentially with more matrices. We already convert to fixed-point on the higher precision kernels in general, so this change is quite simple to the code. If one leaves everything in fixed-point after the conversion, we could have the same algorithm, just with more matrices. For FP64x2, instead of having 4 matrices to cover the 117 bits, one could use 14 or 15 matrices to store the converted form and then a huge number of multiplies. Note that if one splits both $A$ and $B$ into $y$ parts, then one only needs the first $y*(y+1)/2$ multiplies with this type of trick. So if $y=14$, this is 105 multiplies, and $y=15$, this is 128 multiplies. Note also that on some processors, Int8 is many times faster than FP64, so this actually yields better results than reported in our existing performance charts. That is, in Figure~\ref{fig:perf}, we saw our algorithm basically running 12x slower than DGEMM. So, on any processor where Int8 runs $104/12 \approx 9$ times faster than FP64, the resulting code would end up even faster yet. In fact, if BLIS had supported Int8-GEMM when we started this work, we would simply have implemented this algorithm instead and demonstrated performance improvements on GPUs that manage Int8 running at over 9x faster than FP64. 
}

\subsubsection{Mixed precision}

As of this writing, BLIS  supports mixing domains and precision of the four data types (\texttt{s, d, c, z}) supported by the traditional BLAS~\cite{BLIS7} and is being extended to support new precisions (e.g. FP16 and FP64x2). Consider $ C := A B + C $.  Each of $ \{ A, B , C \} $ may have a different precision and/or the computation may be performed in a different precision than some or all of the operands. The current approach in BLIS is to convert all operands to the highest precision that is involved. Ignoring the copy overheads which are lower order in complexity, the time for the mixed-precision GEMM will be the same time as the higher precision GEMM. The presented techniques in this paper can be leveraged to, for example, selectively cascade operands to attain a desired precision while not being limited to the performance of the highest precision computations. For instance, if one can compute the necessary multiplies in lower precision faster than a single multiply in the higher precision (this is clearly data type and machine type dependent), then this new method of doing mixed precision will be more efficient than the current state-of-the-art inside BLIS.  
\remove{
is being extended to support new precisions (e.g.,  FP16 and FP64x2) and mixed precision computation.  Consider $ C := A B $.  Each of $ \{ A, B , C \} $ may have a different precision and/or the computation may be performed in a different precision than some or all of the operands.
The current approach is to convert all operands to the highest precision that is involved.

The presented techniques can be leverage to, for example, selectively cascade operands to attain a desired precision.}

\NoShow{
\subsection{\gemm-like operations}
}

\NoShow{
\subsection{Analyzing numerical aspects of cascading matrix multiplication}
}

\subsection{Targeting GPUs}

As mentioned in Section~\ref{sec:BLISCascadingGemm},
the techniques employed by the CPU implementation of cascading \gemm\ are similar to those used to implement a high-performance Strassen's algorithm.  In~\cite{gpustrassen}, it was shown how there is a parallel between the BLIS \response{implementation} of \gemm\ on CPUs and the CUTLASS implementation of \gemm\ on Nvidia GPUs~\cite{cutlassweb}.  This suggests that the approach may translate naturally to such GPUs.

\NoShow{\subsection{Numerical analysis}

\label{sec:moreNA}

{\bf I suggest we take this whole section out. While what we say still holds, I suggest not drawing attention to it...}  In Section~\ref{sec:analysis}, we analyzed the impact of cascading matrices on the accuracy of a matrix multiplication.
We discussed the issue and give theoretical insights, but did not solve all issues. While we have outlined a way to detect when we are suffering from severe cancellation error at almost no cost, the  question that still remains is whether there is a way to know for sure that a result is more or less accurate than FP64x2. However, the fact that we have insight into cancellation error is already more than can be said about most other \gemm\ implementations.
}

\remove{\subsection*{\devangi{\sout{Detecting cancellation}}}
\label{sec:moreCancellation}
We can remove since we discuss this earlier. Maybe not. I need to go back and read the sections we have earlier to see if we can get rid of this.} 

\subsection{Scaling/balancing methods}

Mathematically,
$ \widehat A \widehat B = 
\widehat A F F^{-1} \widehat B $, where $ F $ is nonsingular.  If $ F $ is diagonal, then this means that columns of $ \widehat A $ are scaled by the corresponding diagonal element of $ F $ and rows of $ \widehat B $ by the inverse of the corresponding diagonal element of $ F $.
Now, if a diagonal element of $ F $ is chosen to be small enough, the contribution of the corresponding column of $ \widehat A $ in a cascaded matrix multiplication can be reduced to FP64 accuracy, since the entries can become small relative to other entries in the same row of $ \widehat A $.  A similar argument can be made for scaling a row of $ \widehat B $.
In other words, how columns of $ \widehat A $ and/or rows of $ \widehat B $ are scaled matters.

This brings up the question \response{of} whether a diagonal matrix $ F $ can be chosen to improve the accuracy of the cascaded matrix multiplication.  We leave answering this to future research.

\NoShow{
\subsection{Converting matrices into chunks}
}

\NoShow{
\subsection{Exploring hardware support for the fixed precision computations}
}

\NoShow{
\subsection{Analyzing data reuse}
    Should the change in reuse of data change the algorithm?
}

\NoShow{
\subsection{Mixing operand precision within cascaded matrix multiplication}
}

\subsection{Analytical model of performance}

An open question remains how to create an analytical model that estimates the performance of the proposed cascading \gemm\ techniques. Such a model would guide the choices of when, in the high performance \gemm\ algorithm, Phase 1 and Phase 3 can be performed to maximize the data reuse.

\remove{\subsection{Hardware support for cascading}
\greg{I'm not sure we still need this section. Yes, I can talk about how neat it'd be to do these splits with a hardware instruction, but I suspect Intel would rather I submit a HW patent on the idea and NOT publish it. I'm not sure it has value in this paper to be honest. Can we discuss next time?}}

\subsection{Supporting other level-3 BLAS-like functionality}
All Level-3 BLAS algorithms can be written in terms of GEMM \cite{LAWN107}. One naive solution then is to take a fast cascading GEMM implementation and re-use it as much as possible during any BLAS-3 routine. But also note that Phase 1 involves a lot of conversion, so integrating Phase 1 into specific calls of say TRSM can make the routine run even faster yet.

\section*{Acknowledgements}

We thank members of the Science of High-Performance Computing (SHPC) group for their encouragement and feedback.  This research was sponsored in part by the National Science Foundation (Award CSSI-2003921).

{\em Any opinions, findings and conclusions or recommendations expressed in this material are those of the author(s) and do not necessarily reflect the views of the National Science Foundation (NSF).}

\bibliographystyle{plain}
\bibliography{biblio}

\newpage

\appendix

\section{Cascading a F64x2: Technical details.}
\label{app:A}

How to cascade a given FP64x2 number into its four parts is conceptually easy, but requires care in practice so that it doesn't create unacceptable overhead.

Recall that cascading a number does not happen in isolation: All elements in a row of $ A $ or column of $ B $ must cascade conformally.  
Thus, we assume that $ \sigma_0 = 2^e $ has been determined where $ e $ equals the exponent of the element with largest magnitude.
The element we are cascading, $ \widehat \chi $ , is first scaled so that
\[
\widehat \chi := \widehat \chi \times 2^{-e} =
\pm .
\begin{array}{|l|l|l|l} \hline
\beta_0 \cdots  \beta_{ \digits_0 -1} & \beta_{ \digits_0 } \cdots \beta_{\digits_1-1} &
\beta_{\digits_1} \cdots \beta_{\digits_2-1} &
\beta_{\digits_2} \cdots ~~~
\\
\hline
\end{array}
.
\]
Here, as before, $ D_0 $, $ D_1 $, and $ D_2 $ indicate where splits occur, as illustrated with the boxes.  
If we are working with the element that is largest in magnitude, then $ \beta_0 = 1 $ and all digits starting with $ \beta_{2D} $ equal zero.  
Otherwise, there may be leading zero bits.  
The question now becomes how to extract the indicated parts into four FP64 numbers.

One way to achieve this is to recognize that ranges of the mantissa can be extracted by masking the corresponding bits and using a bit-wise logical AND operation.  This allows one to extract $ \chi_0 $, $ \chi_1 $, and $ \chi_2 $, after which $ \chi_3 = \widehat \chi - \chi_0 - \chi_1 - \chi_2 $.
The problem with this approach is that leading zeroes in sections may extract denormalized FP64 numbers and in the end FP64x2 subtraction is required to compute $ \chi_3 $.  This results in an algorithm complicated by many conditionals, which in implementation incur considerable overhead.  In other words, our experience is that this is slow.

We instead favor an approach that only computes with FP64 numbers, performance only FP64 arithmetic, and avoids all conditionals.  We have found it to be fast enough.

The first observation is that, because the ranges of bits covered by $ \chi_{\rm hi} $ and $ \chi_{\rm lo} $ do not overlap, 
the two FP64 numbers that store $ \widehat \chi $, $ \chi_{\rm hi} $ and $ \chi_{\rm lo} $ can be cascaded seperately so that
\[
\begin{array}{rcl}
\chi_{\rm hi} &=& 
\chi_{{\rm hi},0}
+
\chi_{{\rm hi},1}
\sigma_1
+
\chi_{{\rm hi},2}
\sigma_2
+
\chi_{{\rm hi},3}
\sigma_3
\\
\chi_{\rm lo} &=& 
\chi_{{\rm lo},0}
+
\chi_{{\rm lo},1}
\sigma_1
+
\chi_{{\rm lo},2}
\sigma_2
+
\chi_{{\rm lo},3}
\sigma_3.
\end{array}
\]
Then
\[
\chi_i =
\chi_{{\rm hi},i} +
\chi_{{\rm lo},i},
~i \in \{ 0,1,2,3\},
\]
an addition that is exact when performed in FP64 arithmetic.
This means we can focus on how to cascade a FP64 number, $ \chi_\star $ where $ \star \in \{ {\rm hi}, {\rm lo} \} $, into four parts.

The second observation is that, given a FP64 number 
$ \chi = 
\pm . \beta_0 \beta_1 \cdots 
\beta_{D-1} \beta_D \cdots $
one can compute, in FP64 arithmetic, 
$ \psi =
\pm . \beta_0 \beta_1 \cdots 
\beta_{d-1}
$
via the steps
\begin{itemize}
\item
$ \psi := \chi $.
\item
$ \psi = \psi + 2^{53-d} + 2^{53-d-1} $.
Since the result is stored in a FP64 number, this removes bits starting with $ \beta_D $.
\item
$ \psi = \psi - 2^{53-d} - 2^{53-d-1} $.
\end{itemize}
(More precisely, this rounds to the first $ d $ digits.\NoShow{At least, if you get the details right and you take Greg's word for it!})

The third observation is that this computation will be performed for many FP64x2 numbers and hence various powers of two can be precomputed.

Putting all this together yields the following steps for cascading $ \widehat \chi $, given that it needs to be scaled by $ 2^e $:
\begin{itemize}
\item
Consider $ \widehat \chi = \chi_{\rm hi} + \chi_{\rm lo} $.
\item
$ c_0 := D_0 $, $ c_1 := D_1 - D_0 $, $ c_2 := D_2 - D_1 $
\item
for $ \star \in \{ {\rm hi}, {\rm lo} \} $:
\begin{itemize}
\item
Scale $\chi_\star = \chi_\star \times 2^{-e} $
\item
$ \chi_{\star,0} := \chi_{\star}$;
$ \chi_{\star,0} := \chi_{\star,0} + 2^{53-c_0} + 2^{53-c_0-1}$;
$ \chi_{\star,0} = \chi_{\star,0} - 2^{53-c_0} - 2^{53-c_0-1}$.
\item
$ \chi_\star := \chi_\star - \chi_{\star,0}$; $ \chi_\star := \chi_\star \times 2^{c_0} $
\item
$ \chi_{\star,1} := \chi_{\star}$;
$ \chi_{\star,1} := \chi_{\star,1} + 2^{53-c_1} + 2^{53-c_1-1}$;
$ \chi_{\star,1} = \chi_{\star,1} - 2^{53-c-1} - 2^{53-c_1-1}$.
\item
$ \chi_\star := \chi_\star - \chi_{\star,1}$; $ \chi_\star := \chi_\star \times 2^{c_1} $
\item
$ \chi_{\star,2} := \chi_{\star}$;
$ \chi_{\star,2} := \chi_{\star,2} + 2^{53-c_2} + 2^{53-c_2-1}$;
$ \chi_{\star,2} = \chi_{\star,2} - 2^{53-c_2} - 2^{53-c_2-1}$.
\item
$ \chi_{\star,3} := \chi_\star - \chi_{\star,2}$;
$ \chi_{\star,3} := \chi_{\star,3} \times 2^{c_2} $.
\end{itemize}
\item
$ \chi_i := \chi_{{\rm hi},i} + \chi_{{\rm lo},i} $ for $ i = 0, 1, 2, 3 $.
\end{itemize}
It is likely that many parts will equal zero because the active part of the FP64x2 number fall outside of range covered by the part.  

A final observation is that because conditionals are avoided in the computation, cascading multiple FP64x2 numbers can be vectorized.

\end{document}